\documentclass[acmsmall, screen, nonacm, letterpaper, 10pt]{acmart}

\AtBeginDocument{%
  }

\setcopyright{none}
\copyrightyear{2018}
\acmYear{2018}
\acmDOI{XXXXXXX.XXXXXXX}

\usepackage{bbm}
\usepackage{mathrsfs}
\usepackage{amsmath}
\usepackage{stmaryrd}
\usepackage{amsthm}

\newcommand{\labelSet}{\mathscr{L}}
\newcommand{\wtypeSet}{\mathscr{W}}
\newcommand{\lcSet}{\mathscr{Q}}
\newcommand{\gateSet}{\mathscr{G}}
\newcommand{\circSet}{\mathit{CRL}}
\newcommand{\mvalSet}{\mathit{BVAL}}
\newcommand{\bvalSet}{\mvalSet} 
\newcommand{\mtypeSet}{\mathit{BTYPE}}
\newcommand{\btypeSet}{\mtypeSet} 
\newcommand{\labOne}{\ell} \newcommand{\labTwo}{k} \newcommand{\labThree}{t} \newcommand{\labFour}{q}
\newcommand{\wtypeOne}{w} \newcommand{\wtypeTwo}{r}
\newcommand{\wtypeMultisetOne}{W}
\newcommand{\lcOne}{Q} \newcommand{\lcTwo}{L} \newcommand{\lcThree}{H} \newcommand{\lcFour}{K}
\newcommand{\gateOne}{g} 
\newcommand{\circuitOne}{\mathcal{C}} \newcommand{\circuitTwo}{\mathcal{D}} \newcommand{\circuitThree}{\mathcal{E}} \newcommand{\circuitFour}{\mathcal{F}}
\newcommand{\mtypeOne}{T}
\newcommand{\btypeOne}{\mtypeOne} 
\newcommand{\mtypeTwo}{U}
\newcommand{\btypeTwo}{\mtypeTwo} 
\newcommand{\cidentity}[1]{id_{#1}}
\newcommand{\gateapp}[3]{#1(#2) \to #3}
\newcommand{\unitt}{\mathbbm{1}}
\newcommand{\qubitt}{\mathsf{Qubit}}
\newcommand{\bitt}{\mathsf{Bit}}

\newcommand{\struct}[1]{\bar{#1}}
\newcommand{\bundle}{\struct} 
\newcommand{\concat}[2]{#1 :: #2}
\newcommand{\discarded}[1]{\operatorname{reusable}(#1)}

\newcommand{\natSet}{\mathbb{N}}

\newcommand{\vsubSet}{\mathit{VSUB}}

\newcommand{\termSet}{\mathit{TERM}}
\newcommand{\indexSet}{\mathit{INDEX}}
\newcommand{\valSet}{\mathit{VAL}}
\newcommand{\typeSet}{\mathit{TYPE}}
\newcommand{\ptypeSet}{\mathit{PTYPE}}
\newcommand{\varOne}{x} \newcommand{\varTwo}{y} 
\newcommand{\ivarOne}{i} \newcommand{\ivarTwo}{j} 
\newcommand{\contextOne}{\Gamma}
\newcommand{\pcontextOne}{\Phi}
\newcommand{\icontextOne}{\Theta} \newcommand{\icontextTwo}{\Xi} \newcommand{\icontextThree}{\Delta}
\newcommand{\termOne}{M} \newcommand{\termTwo}{N} 
\newcommand{\indexOne}{I} \newcommand{\indexTwo}{J} \newcommand{\indexThree}{E} \newcommand{\indexFour}{F}  
\newcommand{\valOne}{V} \newcommand{\valTwo}{W} \newcommand{\valThree}{X} \newcommand{\valFour}{Y} \newcommand{\valFive}{Z} \newcommand{\valSix}{S}
\newcommand{\typeOne}{A} \newcommand{\typeTwo}{B} 
\newcommand{\ptypeOne}{P} \newcommand{\ptypeTwo}{R}
 
\newcommand{\emptycontext}{\bullet}

\newcommand{\emptysub}{\emptyset}
\newcommand{\unitv}{*}
\newcommand{\tuple}[2]{\langle#1,#2\rangle}

\newcommand{\abs}[3]{\lambda #1_{#2}.#3}
\newcommand{\app}[2]{#1\,#2}
\newcommand{\liftoperator}{\operatorname{\mathsf{lift}}}
\newcommand{\lift}[1]{\liftoperator #1}
\newcommand{\force}[1]{\operatorname{\mathsf{force}} #1}
\newcommand{\boxoperator}{\operatorname{\mathsf{box}}}
\newcommand{\boxt}[2]{\boxoperator_{#1} #2}
\newcommand{\boxedCirc}[3]{(#1,#2,#3)}
\newcommand{\applyoperator}{\operatorname{\mathsf{apply}}}
\newcommand{\apply}[2]{\applyoperator(#1,#2)}
\newcommand{\letoperator}{\mathsf{let}}
\newcommand{\letin}[3]{\letoperator\; #1 = #2 \;\mathsf{in}\; #3}
\newcommand{\dest}[4]{\letoperator\; \tuple{#1}{#2} = #3 \;\mathsf{in}\; #4}

\newcommand{\returnoperator}{\mathsf{return}}
\newcommand{\return}[1]{\returnoperator\; #1}
\newcommand{\foldoperator}{\operatorname{\mathsf{fold}}}

\newcommand{\nil}{\mathsf{nil}}
\newcommand{\cons}[2]{\mathsf{cons}\; {#1}\; {#2}}

\newcommand{\lineararrow}{\multimap}
\newcommand{\arrowt}[4]{#1 \multimap^{#3}_{#4} #2}
\newcommand{\bang}[2]{\operatorname{!}^{#1}#2}
\newcommand{\circt}[3]{\operatorname{\mathsf{Circ}}^{#1}(#2,#3)}
\newcommand{\tensor}[2]{#1 \otimes #2}
\newcommand{\listt}[2]{\mathsf{List}^{#1}\,{#2}}

\newcommand{\natOne}{n}  
\newcommand{\iplus}[2]{#1 + #2}
\newcommand{\iminus}[2]{#1 - #2}
\newcommand{\imult}[2]{#1 \cdot #2}
\newcommand{\imax}[2]{\mathsf{max}\left(#1,#2\right)}

\newcommand{\iseq}[2]{#1 \varoast #2}
\newcommand{\ipar}[2]{#1 \varoplus #2}
\newcommand{\iseqBounded}[3]{\scalebox{1.5}{$\iseq{}{}$}_{#1<#2}#3}
\newcommand{\iparBounded}[3]{\scalebox{1.5}{$\ipar{}{}$}_{#1<#2}#3}
\newcommand{\iappend}[1]{\mathsf{append}_{#1}}
\newcommand{\igate}[2]{\mathsf{out}_{#1,#2}}
\newcommand{\iid}[1]{\mathsf{id}_{#1}}
\newcommand{\imempty}{\mathsf{e}}
\newcommand{\iwire}[1]{\mathsf{wire}_{#1}}
\newcommand{\iEqJudgment}[4]{{#2} \vDash_{#1} {#3} = {#4}}
\newcommand{\iLeqJudgment}[4]{{#2} \vDash_{#1} {#3} \leq {#4}}
\newcommand{\imaxsug}[1]{\mathsf{max}(#1)}

\newcommand{\config}[2]{(#1,#2)}
\newcommand{\eval}{\Downarrow}
\newcommand{\freshlabelsfunction}{\operatorname{\mathit{freshlabels}}}
\newcommand{\freshlabels}[1]{\freshlabelsfunction(#1)}
\newcommand{\appendfunction}{\operatorname{\mathit{emit}}}
\newcommand{\append}[3]{\appendfunction(#1,#2,#3)}
\newcommand{\sub}[2]{[#1/#2]}
\newcommand{\isub}[2]{\{#1/#2\}}

\newcommand{\vsubOne}{\gamma}
\newcommand{\isubOne}{\vartheta}  \newcommand{\isubThree}{\delta}
\newcommand{\extension}[2]{[#1 \mapsto #2]}

\newcommand{\freelabels}{\mathit{FL}}
\newcommand{\rcount}[1]{\#(#1)}

\newcommand{\multiset}[1]{\widetilde{#1}}

\newcommand{\PQ}{\textsf{Proto-Quipper}}

\newcommand{\PQR}{\textsf{Proto-Quipper-R}}
\newcommand{\PQRA}{\textsf{Proto-Quipper-RA}}
\newcommand{\Haskell}{\texttt{Haskell}}

\newcommand{\Quipper}{\texttt{Quipper}}
\newcommand{\Qiskit}{\texttt{Qiskit}}
\newcommand{\Cirq}{\texttt{Cirq}}

\newcommand{\qura}{\texttt{QuRA}}
\newcommand{\qasm}{\texttt{QASM}}
\newcommand{\crl}{\textsf{CRL}}

\newcommand{\rcmOne}{\mathscr{R}}

\newcommand{\carrierOne}{A}
\newcommand{\rcmFunOne}{\mu}
\newcommand{\rcmGatecount}{\mathbf{Gatecount}}
\newcommand{\gatecount}{\operatorname{gatecount}}
\newcommand{\rcmWidth}{\mathbf{Width}}
\newcommand{\width}{\operatorname{width}}
\newcommand{\rcmDepth}{\mathbf{Depth}}
\newcommand{\depth}{\operatorname{depth}}

\newcommand{\icsOne}{\mathfrak{c}}
\newcommand{\icsId}[2]{\mathfrak{id}_{#1}^{#2}}
\newcommand{\icsGate}[3]{\mathfrak{out}_{#1}^{#2,#3}}
\newcommand{\icsAppend}[2]{\mathfrak{append}_{#1}^{#2}}
\newcommand{\icsJudgment}[6]{#2\vdash_{#1}^{s} {#4} : {#3} \to {#5} ; {#6}}
\newcommand{\bundleJudgment}[5]{#2:#3\vdash_{#1}^b #4:#5}
\newcommand{\icsWFJudgment}[2]{#1 \vdash #2}
\newcommand{\gatescheme}[1]{\operatorname{GTS}_{#1}}
\newcommand{\ginputs}[1]{\operatorname{inputs}(#1)}

\newcommand{\icsGatecount}{\mathbf{Gatecount}}
\newcommand{\icsWidth}{\mathbf{Width}}
\newcommand{\icsDepth}{\mathbf{Depth}}

\newcommand{\itesOne}{\mathfrak{t}}

\newcommand{\itesMempty}[1]{\mathfrak{e}_{#1}}
\newcommand{\itesWires}[2]{\mathfrak{wire}_{#1}^{#2}}
\newcommand{\itesSeq}[1]{\mathfrak{seq}_{#1}}
\newcommand{\itesPar}[1]{\mathfrak{par}_{#1}}
\newcommand{\itesCJudgment}[8]{#2;#3:#4;#5 \vdash_{#1}^{c} #6:#7;#8}
\newcommand{\itesVJudgment}[7]{#2;#3:#4;#5 \vdash_{#1}^{v} #6:#7}
\newcommand{\itesIWFJudgment}[2]{#1 \vdash #2}
\newcommand{\itesTWFJudgment}[3]{#1;#2 \vdash #3}
\newcommand{\itesSubJudgment}[5]{#2;#3 \vdash_{#1} #4 <: #5}

\newcommand{\wirecontent}[1]{wc(#1)}
\newcommand{\abswirecontent}[3]{wc(#1;#2;#3)}

\newcommand{\resourcecontent}[1]{\#(#1)}
\newcommand{\blistt}[3]{\textsf{List}_{#1<#2}\,#3}
\newcommand{\mcirct}[4]{\operatorname{\mathsf{Circ}}^{#1}_{#2}(#3,#4)}
\newcommand{\rcons}[2]{\mathsf{rcons}\; {#1}\; {#2}}
\newcommand{\mfold}[4]{\foldoperator_{#1}\,{#2}\,{#3}\,{#4}}
\newcommand{\mboxt}[3]{\boxoperator_{#1,#2} #3}
\newcommand{\itesGatecount}{\mathbf{Gatecount}}
\newcommand{\itesWidth}{\mathbf{Width}}

\newcommand{\labels}{\operatorname{\mathit{labs}}}
\newcommand{\vcorrect}[5]{\mathit{LRC}_{#1}^{#2}(#3,#4,#5)}
\newcommand{\vsemint}[3]{\mathfrak{V}_{#1}^{#2}[#3]}
\newcommand{\tsemint}[4]{\mathfrak{T}_{#1}^{#2}[#3;#4]}
\newcommand{\csemint}[3]{\mathfrak{C}_{#1}^{#2}[#3]}
\newcommand{\isemint}[2]{\mathfrak{I}_{#1}[#2]}
\newcommand{\semCJudgment}[8]{#2;#3:#4;#5 \vDash_{#1}^{c} #6:#7;#8}
\newcommand{\semVJudgment}[7]{#2;#3:#4;#5 \vDash_{#1}^{v} #6:#7}

\renewcommand{\iff}{\Leftrightarrow}
\renewcommand{\implies}{\Rightarrow}

\usepackage{adjustbox}
\usepackage{quantikz}
\usepackage{mathpartir}
\usepackage{listings}
\lstdefinestyle{myStyle}{
	belowcaptionskip=1\baselineskip,
	breaklines=true,
	numbers=left,
  lineskip=-2pt,
	basicstyle=\footnotesize\ttfamily,
	keywordstyle=\bfseries\color{violet!50!black},
	commentstyle=\itshape\color{green!30!black},
	identifierstyle=\color{blue}
}
\lstdefinestyle{myStyleEmbedded}{
	belowcaptionskip=1\baselineskip,
	breaklines=true,
	numbers=none,
  lineskip=-2pt,
	basicstyle=\footnotesize\ttfamily,
	keywordstyle=\bfseries\color{violet!50!black},
	commentstyle=\itshape\color{green!30!black},
	identifierstyle=\color{blue}
}
\lstdefinestyle{myStyleEmbeddedPlain}{
	belowcaptionskip=1\baselineskip,
	breaklines=true,
	numbers=none,
  lineskip=-2pt,
	basicstyle=\footnotesize\ttfamily,
	keywordstyle=\normalshape
}

\begin{document}

\title{Flexible Type-Based Resource Estimation in Quantum Circuit Description Languages}

\author{Andrea Colledan}
\email{andrea.colledan@unibo.it}
\orcid{0000-0002-0049-0391}
\affiliation{%
	\institution{University of Bologna}
	\city{Bologna}
	\country{Italy}
}
\affiliation{%
	\institution{INRIA Sophia Antipolis}
	\city{Valbonne}
	\country{France}
}

\author{Ugo Dal Lago}
\email{ugo.dallago@unibo.it}
\orcid{0000-0001-9200-070X}
\affiliation{%
	\institution{University of Bologna}
	\city{Bologna}
	\country{Italy}
}
\affiliation{%
	\institution{INRIA Sophia Antipolis}
	\city{Valbonne}
	\country{France}
}

\renewcommand{\shortauthors}{Colledan and Dal Lago}

\begin{abstract}
  We introduce a type system for the \Quipper\ language designed to derive upper bounds on the size of the circuits produced by the typed program. This size can be measured according to various metrics, including \emph{width}, \emph{depth} and \emph{gate count}, but also variations thereof obtained by considering only some wire types or some gate kinds. The key ingredients for achieving this level of flexibility are effects and refinement types, both relying on \emph{indices}, that is, generic arithmetic expressions whose operators are interpreted differently depending on the target metric. The approach is shown to be correct through logical predicates, under reasonable assumptions about the chosen resource metric. This approach is empirically evaluated through the \qura\ tool, showing that, in many cases, inferring tight bounds is possible in a fully automatic way.
\end{abstract}




\maketitle

\section{Introduction}

Since its introduction, the quantum computing paradigm has promised to have a disruptive impact on many areas of computing, ranging from cryptography \cite{shor,regev} to machine learning \cite{wittek}. If, on the one hand, these promises have found clear confirmation on the side of the underlying computational model, i.e. that of quantum circuits, on the other hand the implementation of quantum circuits with acceptable error rates is still considered a difficult task \cite{quantum-challenges,nisq}, even when the amount of quantum data or the number of operations to be executed is relatively small. Despite the great progress recorded in the last five years \cite{osprey,jiuzhang,sycamore}, it seems safe to say that qubits and the ability to perform quantum operations on them are, unlike their classical relatives, very precious resources, which need to be accounted for with great care: as soon as the number of qubits gets nontrivial, coherence problem start to arise \cite{quantum-decoherence}, and the very number of operations that one can perform on qubits before their state becomes too noisy is very low. 

Quantum architectures thus seems to be intrinsically more error-prone than their classical 
counterparts. In spite of this, they are still accessed through high-level 
programming languages, and proposals in this direction 
abound \cite{survey-gay,survey-selinger}. Most quantum programming languages 
are somehow hybrid and allow \emph{both} classical \emph{and} quantum computing 
to happen within the same program. It is thus the job of the compiler, or of 
the interpreter, or of the programmer itself, to split the workload between 
classical, slow, but fault-tolerant hardware and quantum, fast, but faulty 
devices. Some of the  existing programming languages are designed around the 
QRAM model \cite{qram}, in which control is classical and data are quantum. Others, 
instead, take the form of \emph{circuit description languages}, that is, 
programming languages and libraries in which programs are not meant to directly 
execute quantum operations, but rather implement circuit-building 
algorithms, whose output circuits represent the quantum computations later 
executed on a quantum device.

In both QRAM-based and circuit description languages, having a precise idea of how 
many qubits and how many quantum operations are needed to solve a generic 
instance of a computational problem is of paramount importance, and the 
corresponding verification problem has recently received the attention of the 
programming language community~\cite{quantum-weakest,proto-quipper-r}. One might imagine 
this kind of analysis to closely resemble the classical counterpart, and 
rightly so. However, some key differences are at play. Of course, the 
underlying computational model is different, and this needs to be taken into 
account. But there is more: even if we focus on the quantum circuit model 
only, existing hardware architectures are anything but homogeneous, and 
different devices have their own cost-related strengths and weaknesses. For 
example, in some architectures the operation of swapping between qubits is 
extremely expensive, while in others it is not. As another example, circuits 
often contain both classical and quantum wires, and while keeping the latter 
under control is of paramount importance, the same cannot be said about the 
former. In other words, any resource analysis framework for quantum circuits 
should be designed so as to be \emph{as flexible as possible}, accomodating for 
different resource metrics and for varying definitions of the same metric.

In this paper we go precisely towards a flexible and scalable methodology for 
quantum resource estimation, by defining a family of type systems for the 
\Quipper\ language where each type system is capable of deriving upper 
bounds on the size of the circuits produced by a program with respect to a 
specific metric. In fact, all of these type systems can be seen as instances of 
\emph{one} type system, called \PQRA, in which some operators occurring in 
types and dealing with resource estimation are left uninterpreted. The 
aforementioned family is thus indexed on a \emph{resource metric 
interpretation} (RMI), that is, an interpretation of these symbols. Remarkably, 
an instance of \PQRA\ is provably correct with respect to the chosen metric 
once certain reasonable inequalities between the RMI and the 
metric's ground truth are shown to hold.
When coming up with our approach, we were inspired by the recent work of Colledan and 
Dal Lago on width estimation \cite{proto-quipper-r}, substantially 
generalizing it in order to ensure that a wide range of circuit metrics can be 
dealt with: not only \emph{global} metrics such as width, but also metrics that 
are by their nature \emph{local} to wires, such as the depth of a circuit.

Technically, \PQRA\ is based on two fundamental ingredients, namely refinement types \cite{refinement-haskell,refinement-ml} and effects\cite{effect-systems-revisited}. The former allow us to attach relevant quantitative 
information to wire types, i.e. types that represent data in the underlying 
quantum circuit, and thus allow for local metrics, such as depth, to be 
treated. Effects, on the other hand, attribute to each circuit-building 
computation a description of the size of the resulting circuit \emph{as a whole}, 
and as such are suitable for managing global metrics, as already shown by 
Colledan and Dal Lago. What ties the two types of metrics together as the 
vehicle for expressing bounds is a language of \emph{index terms}, i.e. 
arithmetic expressions that contain natural-valued variables, and as such can 
express upper bounds that depend on circuit-building parameters. Index terms 
are well-known to provide a flexible and powerful methodology for injecting arithmetic 
reasoning within a type 
system~\cite{linear-dependent-types-relative-completeness,linear-dependent-types-cbv,linear-dependent-types-privacy,geometry-of-types}.

The most important results of this work are the following:
\begin{itemize}
	\item
	On the one hand, the proof that the definitional apparatus based on \PQRA\ 
	and RMIs is, under mild assumptions, correct by construction. In other 
	words, that the derivable type judgments indeed provide correct upper bounds to the corresponding circuit metric. The proof proceeds by adapting the 
	logical predicates technique to \PQRA.
	\item
	On the other hand, the implementation of the generic \PQRA\ framework as the core of \qura, a tool for the resource verification of circuit description languages. By bringing \PQRA's genericity to \qura, we were able to effortlessly extend its functionality to the verification of gate count, T-count, and several variations of these and other circuit metrics.
\end{itemize}

The rest of the paper is structured as follows. Section \ref{sec:overview} introduces some core concepts underpinning this work: quantum circuits, their size, \Quipper\ and the problem of estimating the size of the circuits built by circuit description languages. Section \ref{sec:the-proto-quipper-language} then provides a brief, informal introduction to the \PQ\ family of calculi \cite{proto-quipper-s, proto-quipper-m,proto-quipper-r}, which our formalism is based on. Section \ref{sec:quantumcircuitmetrics} picks up where Section \ref{sec:overview} left off, formalizing the notions of circuit and circuit metric and introducing a simple parametric type system for the estimation of the size of individual circuits. Section \ref{sec:pqra} raises the level of abstraction, introducing \PQRA, a circuit description language with a type system capable of estimating different circuit metrics depending on an underlying semantic interpretation of resources. \PQRA\ is proven to be correct under mild assumptions in Section \ref{sec:correctness}. Section \ref{sec:implementation} describes the efforts to implement \PQRA\ within \qura. In sections \ref{sec:implementation} and \ref{sec:examples} we show that \PQRA\ can be used to analyze the Quantum Fourier Transform algorithm both in terms of global and local resource consumption. Lastly, sections \ref{sec:related-work} and \ref{sec:conclusion} present the related literature, conclusions, and future work.

\section{Resource Estimation in Circuit Description Languages: an Overview}
\label{sec:overview}

This section is dedicated to giving an overview of the resource estimation 
problem in quantum circuit description languages, and particularly in \Quipper. 
Programs written in these languages use both classical resources, necessary for 
the \emph{construction} of a quantum circuit, and quantum resources, 
necessary for its \emph{execution}. Given their scarcity, quantum resources are the ones we are interested in keeping under control. These means giving bounds on the \emph{size} of circuits. But what do we mean by that?

\subsection{Quantum Circuits and their Size}

Quantum circuits are the standard computational model of quantum computation. A circuit consists of an ordered sequence of \textit{gates} applied to a collection of 
\textit{wires}. The latter represent the individual bits and qubits involved in a 
computation, while gates represent the elementary operations applied to them, such as unitary 
transformations, measurements, initializations, and so on. Quantum circuits are often represented graphically, like in Figure \ref{fig:quantum-circuit}: time flows from left to 
right, horizontal lines represent either qubit (when single) or bit (when double) wires, and the various symbols that lie on them represent 
different operations applied to the corresponding qubits or bits.
Any wire starting from $\ket{\phi}$ is initialized as part of the computation to the known state $\phi$, while all other wires are input to the circuit. Any wire that ends with the ground symbol is discarded as part of the computation, while the others contain the circuit's output.

\begin{figure}
	\fbox{
	\begin{quantikz}[row sep = 4pt]
		 	& & 					& 			& \ctrl{1} & \gate{H} & \meter{} & \setwiretype{c}& \ctrl{2} & \ground{}
		\\
		\setwiretype{n} & \lstick{$\ket{0}$} & \setwiretype{q} &  \targ{} & \targ{} &  & \meter{} & \ctrl{1}\setwiretype{c} & & \ground{}
		\\
		\setwiretype{n} & \lstick{$\ket{0}$} & \gate{H}\setwiretype{q} & \ctrl{-1} & & & & \gate{X} & \gate{Z} & 
		\end{quantikz}
	}
	\caption{An example of quantum circuit: the quantum teleportation circuit}
	\label{fig:quantum-circuit}
\end{figure}
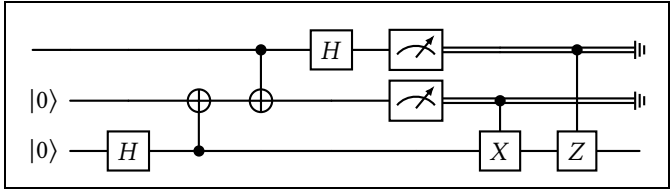

When faced with the question ``what size is a circuit'', there are many ways in which one can answer.
One can first of all consider a measure of \emph{space}, namely the \textit{width} of the circuit. This is defined as the maximum number of wires active at any point in time. By giving an upper bound $n$ to the width of a circuit, we can be sure that no more than $n$ qubits and bits are necessary to execute the circuit, i.e. that any hardware architecture providing at least $n$ logical bits and qubits can evaluate the circuit on any given input.
What if we are rather interested in keeping track of \emph{time}? In this case it makes more sense to consider the \textit{depth} of the circuit, namely the maximum number of gates encountered on any path from an input wire or a wire creation to an output wire or a wire discarding. The resulting metric nicely captures a form of parallel time complexity. Keeping the depth under control is helpful in all those scenarios in which quantum data get progressively less reliable as time passes, or as more operations are applied to them.
A third complexity measure, somehow a mixture of the previous two, is the so 
called \emph{gate count} of a circuit, simply defined as the total number of gates occurring in it. This gives a measure of the total number of 
operations applied to the quantum and classical data during the execution of the circuit. For example, the circuit in Figure 
\ref{fig:quantum-circuit} has a gate count of $8$, a width of $3$ and a depth 
of $6$.

\subsection{From Circuits to \Quipper}

Quantum circuits can be described gate-by-gate, either graphically, like we did in Figure \ref{fig:quantum-circuit}, or with the use of specific languages (like \qasm\ \cite{qasm}). However, because they have fixed input size, individual quantum circuits cannot truly implement algorithms (such as Shor's \cite{shor} or Grover's \cite{grover}), which are by definition capable of dealing with arbitrarily large inputs.
Rather, quantum algorithms are represented by \textit{circuit families}, which contain a different circuit for each input size. But how do we describe circuit families? This is where circuit description languages come into play: by manipulating already built circuits and combining them together in a way that \emph{depends} on the value of one or more (classical) parameters, a single circuit description program can construct an entire family of circuits, and thus implement an algorithm.
This approach to circuit description is nowadays adopted by many high-level quantum circuit description languages and libraries, such as \Qiskit\ \cite{qiskit}, \Cirq\ \cite{cirq} and \Quipper\ \cite{quipper}.

\Quipper, in particular, allows programmers to describe quantum circuits in a very high-level fashion. Besides the gate-by-gate approach, \Quipper\ also supports parametric and hierarchical circuits, promoting a view in which circuits are first-class citizens. \Quipper\ has been shown to be scalable, that is to say, it has been shown to be able to describe complex quantum algorithms that, depending on input parameters, translate to circuits possibly involving trillions of gates applied to millions of qubits \cite{quipper-scalable}.

The already mentioned paper by Colledan and Dal Lago~\cite{proto-quipper-r} introduces a novel methodology by which \Quipper\ programs can be statically analyzed as for the width of the circuits they produce, all this parametrically on the input size. However, notions of size like depth and gate count simply \emph{cannot} be dealt with in this framework. Moreover, the considered notion of width cannot be tuned in any way (e.g. so as to ignore bits). In a sense, our work can be seen as a nontrivial generalization of Colledan and Dal Lago's technique, in which not only depth and gate count, but in fact \emph{any} metric satisfying certain simple conditions can be analyzed.

This section showcases a number of \Quipper\ programs, in order to illustrate how refinements and effects can be used to keep track of the resource information that we are interested in. In doing so, we will find that width and gate count have to be treated in a fundamentally different way than depth. We avoid any formality for now, leaving it for the forthcoming sections.

Let us start with the example of Figure \ref{fig:dumbnot}. The \Quipper\ function on the left builds the quantum circuit on the right: a dumb and expensive implementation of the quantum not operation. The \texttt{dumbNot} function implements negation using a controlled not gate and a temporary qubit \texttt{a}, which is initialized and discarded within the body of the function. This qubit does not appear in the interface of the circuit, but it clearly has an impact on its overall size. Recalling what we said previously, we easily see that this circuit has width $2$, gate count $1$ and depth $1$.

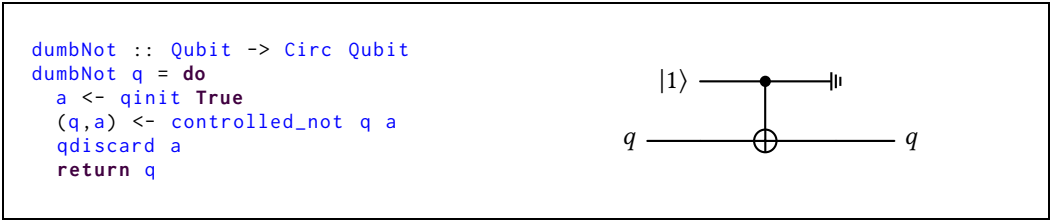
\begin{figure}
	\centering
	\adjustbox{padding=10pt,minipage=\linewidth-20pt,frame}{
	\begin{minipage}{0.49\textwidth}
		\begin{lstlisting}[language=Haskell, style=myStyleEmbedded]
dumbNot :: Qubit -> Circ Qubit
dumbNot q = do
  a <- qinit True
  (q,a) <- controlled_not q a
  qdiscard a
  return q
		\end{lstlisting}
	\end{minipage}
	\begin{minipage}{0.5\textwidth}
		\centering
		\begin{quantikz}[column sep=2em]
			\setwiretype{n}& \lstick{$\ket{1}$} & \ctrl{1} \setwiretype{q} & \ground{} \\
			\lstick{$q$} && \targ & && \rstick{$q$}
		\end{quantikz}
	\end{minipage}
	}
\caption{A \Quipper\ function implementing quantum negation, and the circuit it builds}
\label{fig:dumbnot}
\end{figure}

Consider now the higher-order function in Figure \ref{fig:iter}. This function takes as input a circuit building function \texttt{f}, an integer \texttt{n} and describes the circuit obtained by applying \texttt{f}'s circuit \texttt{n} times to the input qubit \texttt{q}. Knowing the size of the circuit built by \texttt{dumbNot}, what is the size of the circuit built by \texttt{iter dumbNot n}? From a width perspective, it is easy to see that the composition in sequence obeys some form of \textit{maximum} semantics: outputs become inputs and discarded wires are reused when new ones are initialized. Thus, the resulting circuit has still width $2$. What about its gate count and depth? From their perspective, the composition in sequence clearly obeys a form of \textit{addition} semantics, although for different reasons: in the first case, the gate count of two circuits composed together in any way is the sum of their individual gate counts, whereas in the case of depth, two circuits stringed together on the same path have depth equal to the sum of their indivdual depths. Therefore, the circuit built by \texttt{iter dumbNot n} has gate count and depth both equal to $n$.

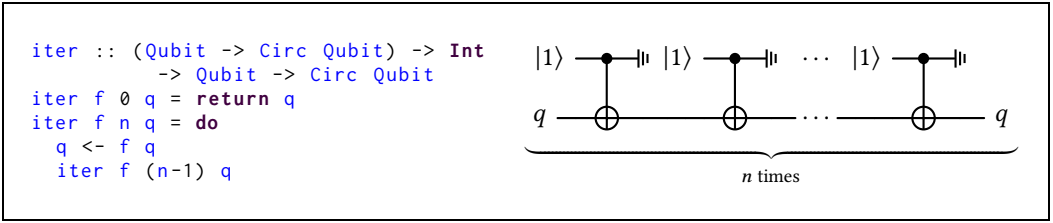
\begin{figure}
	\adjustbox{padding=10pt,minipage=\linewidth-20pt,frame}{
	\centering
	\begin{minipage}{0.49\linewidth}
		\begin{lstlisting}[language=Haskell, style=myStyleEmbedded]
iter :: (Qubit -> Circ Qubit) -> Int
          -> Qubit -> Circ Qubit
iter f 0 q = return q
iter f n q = do
  q <- f q
  iter f (n-1) q
		\end{lstlisting}
	\end{minipage}
	\begin{minipage}{0.5\linewidth}
		\centering
		$\underbrace{
			\begin{quantikz}[column sep=0.7em]
				\setwiretype{n}& \lstick{$\ket{1}$} & \ctrl{1} \setwiretype{q} & \ground{}  & \setwiretype{n} && \lstick{$\ket{1}$} & \ctrl{1} \setwiretype{q} & \ground{} & \setwiretype{n} \ \ldots \ &&&\lstick{$\ket{1}$} & \ctrl{1} \setwiretype{q} & \ground{} \\
				\lstick{$q$} && \targ{} & &&& &\targ{}&& \ \ldots \ &&&& \targ{}&& \rstick{$q$}
			\end{quantikz}
		}_{n \text{ times}}$
	\end{minipage}
}
	\caption{A higher-order \Quipper\ function and the result of its application to \texttt{dumbNot} from Figure \ref{fig:dumbnot}}
	\label{fig:iter}
\end{figure}

\subsubsection{A Type-based Analysis}
Intuitively, this kind of reasoning seems natural, but how can we draw the same conclusions by means of static analysis? We are facing two major hurdles here: higher-order types and input-dependent size. As for the former, we closely follow Colledan and Dal Lago's approach. In their contribution, the arrow type is annotated both with the size of the circuit produced by the corresponding function once applied (a standard technique in effect typing \cite{types-and-effects}) and with information about the wires enclosed in the function's closure (in a way reminiscent of closure types \cite{closure-types}). This last piece of information takes the form of a composite type and is essential to guarantee the correctness of the analysis even in the case of data hiding, since wires have a size, even by themselves. The \texttt{dumbNot} function does not capture anything from its environment, so it has type $\qubitt \to^2_{\unitt} \qubitt$ or $\qubitt \to^1_{\unitt} \qubitt$, depending on the resource being considered.

This brings us to input dependency. It is clear that the aforementioned size 
annotation cannot be a mere \textit{constant}, but should rather \emph{depend} 
on the parameters of the circuit-building function, in this case the number $n$ that we pass to \texttt{iter}. We solve this problem by allowing size annotations to contain 
natural-valued variables, which correspond to these classical parameters.
All in all, knowing that \texttt{dumbNot} has width-annotated type $\qubitt \to^2_\unitt \qubitt$, the system is capable of (correctly) attributing type $\qubitt \to^2_\unitt \qubitt$ to the \texttt{iter dumbNot n} term, by trivially taking the maximum of $n$ copies of $2$. Similarly, knowing that \texttt{dumbNot} has gate-count-annotated type $\qubitt \to^1_\unitt \qubitt$, the system is capable of giving type $\qubitt \to^n_\unitt \qubitt$ to \texttt{iter dumbNot n}, by adding up $n$ copies of $1$. The same goes for depth.

\begin{figure}
	\adjustbox{padding=10pt,minipage=\linewidth-20pt,frame}{
		\centering
		\begin{minipage}{0.5\linewidth}
			\begin{lstlisting}[language=Haskell, style=myStyleEmbedded]
mapHadamard :: [Qubit] -> Circ [Qubit]
mapHadamard [] = return []
mapHadamard (q:qs) = do
  q <- hadamard  q
  qs <- mapHadamard qs
  return (q:qs)
			\end{lstlisting}
		\end{minipage}
		\begin{minipage}{0.5\linewidth}
			\centering
			\begin{quantikz}[column sep=1em, row sep=8pt]
				\lstick{$q_1$} & \gate{H} & & & \rstick{$q_1$}\\
				\lstick{$q_2$} & & \gate{H} & & \rstick{$q_2$}\\
				\wave&&&&\\
				\lstick{$q_n$} & & &\gate{H} & \rstick{$q_n$}
			\end{quantikz}
		\end{minipage}
	}
	\caption{A \Quipper\ function with variable input size, and the circuit it builds on an input of size $n$}
	\label{fig:mapHadamard}
\end{figure}
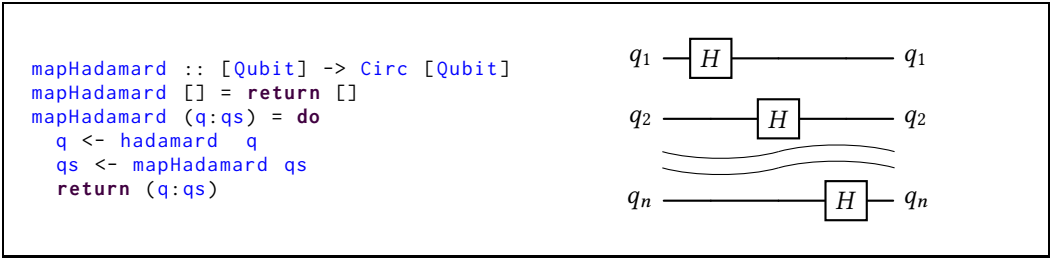

Let us now consider a slightly different example. Figure \ref{fig:mapHadamard} shows a \Quipper\ function that returns a circuit in which the Hadamard gate is applied to $n$ qubits, a common preprocessing step in many quantum algorithms. Knowing that by itself the Hadamard gate has width, gate count and depth equal to $1$, what size is the circuit produced by \texttt{hadamardN} when we apply it to a list of $n$ qubits? This is slightly more complicated than the previous example: at each step we apply a single Hadamard gate, while $n-1$ qubits flow alongside it. As we mentioned eralier, wires have a size of their own. Specifically, a width of $1$ and (trivially) a gate count and a depth of $0$. A bundle of wires of width $n-1$ thus flows \textit{in parallel} to each Hadamard gate. This is to all effects a composition in parallel of two circuits. How do our size metrics behave in this scenario? From the point of view of width, composition in parallel clearly obeys an addition semantics. Gate count also obeys an addition semantics, while depth follows a maximum semantics, but this is fairly irrelevant in this case. We therefore know that each step of \texttt{mapHadamard} has width $n$, gate count $1$ and depth $1$. By iterating $n$ steps in sequence following the same logic from the previous example, we get that \texttt{mapHadamard} applied to a list of $n$ qubits produces a circuit of width $n$, gate count $n$ and depth $n$. Using the same notation for sized lists adopted in \cite{proto-quipper-r}, we can write the type of \texttt{mapHadamard} as $\listt{n}{\qubitt} \to^n_\unitt \listt{n}{\qubitt}$.

\subsubsection{Global and Local Metrics}
Immediately, it is clear that something is not quite right: while we got exact estimates for width and gate count, the estimate for depth is a gross overapproximation. Recall that depth is defined as the maximum number of gates encountered \emph{on any path} from an input to an output of the circuit. It is obvious that every such path in the circuit in Figure \ref{fig:mapHadamard} has only $1$ gate on it, regardless of $n$. However, an overapproximation arises because we have abstracted away information about \emph{which output wires} are matched to \emph{which input wires}, so that when we sequence the $n$ steps, we must assume the worst case scenario, i.e. that the operations of each step are applied to the same wire.

This example highlights a fundamental distinction between metrics such as width and gate count and metrics such as depth. The former are in a way properties of the circuit as a whole, they are \emph{global} and thus amenable to good notions of composition in sequence and parallel. On the other hand, depth is a property of \emph{wires}, rather than \emph{circuits}, and can only be composed adequately at the scale of the individual qubit or bit. The traditional definition of depth of a circuit is thus just an aggregation of the depth of its wires.
This knowledge motivates the most original contribution of this work: instead of keeping track of depth (or other local metrics) via effects, like we do with width and gate count, we do so by directly \emph{refining wire types with depth information}. Let us momentarily ignore global metrics. Let $\qubitt^i$ be the type of a qubit wire at depth $i$ and let \texttt{hadamard} have type $\qubitt^i \to \qubitt^{i+1}$, meaning that the application of a Hadamard gate to a qubit increases its depth by $1$. Then, we can easily derive that \texttt{mapHadamard} has type $\listt{n}{\qubitt^i} \to \listt{n}{\qubitt^{i+1}}$. If we choose $i=0$ and take the max of $n$ copies of $i+1=1$, we correctly conclude that the depth of the circuit built by the function is $1$.
\section{The \PQ\ Language}
\label{sec:the-proto-quipper-language}

As a language, \Quipper\ is embedded in \Haskell, and because of this it lacks a proper formal semantics. However, over the years, a number of calculi have been developed in order to formalize fragments \cite{proto-quipper-s,proto-quipper-m} and extensions \cite{proto-quipper-d,proto-quipper-l,proto-quipper-dyn,proto-quipper-k} of the \Quipper\ language. These calculi are referred to collectively as the \PQ\ language family.
In its most basic form, \PQ\ is a linear lambda calculus equipped with bespoke constructs to build circuits. Circuits are built as a side-effect of a computation, behind the scenes, but more importantly they can also be manipulated as data within the language.

\begin{figure}
	\centering
	\resizebox{!}{!}{\fbox{
		\begin{tabular}{l l r l}
			Types	& $\typeSet$	& $\typeOne,\typeTwo$ & $
			::= \unitt
			\mid \wtypeOne
			\mid \bang{}{\typeOne}
			\mid \tensor{\typeOne}{\typeTwo}
			\mid \typeOne \lineararrow \typeTwo
			\mid \listt{}{\typeOne}
			\mid \circt{}{\mtypeOne}{\mtypeTwo}$\\
			Parameter types \hspace{5pt}		& $\ptypeSet$	& $\ptypeOne,\ptypeTwo$ & $
			::= \unitt
			\mid \bang{}{\typeOne}
			\mid \tensor{\ptypeOne}{\ptypeTwo}
			\mid \listt{}{\ptypeOne}
			\mid \circt{}{\mtypeOne}{\mtypeTwo}$\\
			Bundle types	& $\btypeSet$	& $\mtypeOne,\mtypeTwo$ & $
			::= \unitt
			\mid \wtypeOne
			\mid \tensor{\mtypeOne}{\mtypeTwo}$\\
      \\
      Terms	& $\termSet$	& $\termOne,\termTwo$ & $ ::= \app{\valOne}{\valTwo}
			\mid \dest{\varOne}{\varTwo}{\valOne}{\termOne}
			\mid \force{\valOne}
			\mid \boxt{\mtypeOne}{\valOne}$\\&&&$
			\mid \apply{\valOne}{\valTwo}
			\mid \return{\valOne}
			\mid \letin{\varOne}{\termOne}{\termTwo}$\\
			Values	& $\valSet$	& $\valOne,\valTwo$ & $
			::= \unitv
			\mid \varOne
			\mid \labOne
			\mid \abs{\varOne}{\typeOne}{\termOne}
			\mid \lift{\termOne}
			\mid \boxedCirc{\struct\labOne}{\circuitOne}{\struct{\labTwo}}
			\mid \tuple{\valOne}{\valTwo}$\\&&&$
			\mid \nil
			\mid \cons{\valOne}{\valTwo}
			\mid \mfold{}{\valOne}{\valTwo}{\valThree}$\\
			Wire bundles \hspace{5pt}	& $\bvalSet$	& $\struct\labOne,\struct\labTwo$ & $
			::= \unitv
			\mid \labOne
			\mid \tuple{\struct\labOne}{\struct\labTwo}$
		\end{tabular}
	}}
	\caption{Types and syntax of \PQ}
	\label{fig:pq-types}
\end{figure}

The grammars for \PQ\ types and terms are given in Figure \ref{fig:pq-types}. Speaking at a high level, we can say that \PQ\ employs a linear-nonlinear typing discipline. In particular, wire types ($\qubitt$ and $\bitt$) and arrows are treated linearly. A subset of types, called \textit{parameter types}, represents the values of the language that can be copied. As customary in linear logic, any term of type $\typeOne$ can be \textit{lifted} into a parameter type $\bang{}{\typeOne}$ provided it does not consume linear resources.

Let us focus on the language of terms, and specifically on its domain-specific constructs. On the side of values, we have \textit{labels} and \textit{boxed circuits}.
A label $\labOne$ represents a reference to an output wire of the circuit currently being built and it is unsurprisingly assigned a wire type. Ordered structures of labels form a subset of values called \textit{wire bundles}, which are given \textit{bundle types}.
On the other hand, a boxed circuit $\boxedCirc{\struct\labOne}{\circuitOne}{\struct\labTwo}$ represents a circuit $\circuitOne$ as a datum within the language, together with its input and output wires, given as wire bundles $\struct\labOne$ and $\struct\labTwo$. A boxed circuit is given type $\circt{}{\mtypeOne}{\mtypeTwo}$, where $\mtypeOne$ and $\mtypeTwo$ are respectively the input and output types of the circuit. Boxed circuits can be duplicated and manipulated (e.g. reversed) by primitive functions, but more importantly they can be appended to the underlying circuit via the $\applyoperator$ operator.
This is precisely how circuits are built in \PQ: the $\applyoperator$ operator takes as first argument a boxed circuit $\boxedCirc{\struct\labOne}{\circuitOne}{\struct\labTwo}$ and appends $\circuitOne$ to the underlying circuit $\circuitTwo$. The second argument of $\applyoperator$, a bundle of wires $\struct\labThree$ coming from the free output wires of $\circuitTwo$, tells us \emph{where exactly} to append $\circuitOne$ to $\circuitTwo$.

We expect the language to be equipped with constant boxed circuits corresponding to fundamental circuit operations (e.g. Hadamard, measurement, etc.), but the programmer can also define their own custom circuits via the $\boxoperator$ operator. Intuitively, $\boxoperator$ takes as input a circuit-building function $f$ and executes it in a sandboxed environment, isolated from the current circuit. Therefore, $f$ produces a standalone circuit $\circuitOne$, which is returned by $\boxoperator$ as a boxed circuit of the form $\boxedCirc{\struct\labOne}{\circuitOne}{\struct\labTwo}$.



On the classical side of things, note that \PQ\ does \textit{not} support general recursion. Instead, a limited form of recursion on lists is provided in the form of a primitive $\foldoperator$ constructor, which takes as argument a (duplicable) step function $\valOne$ of type $\bang{}{((\tensor{\typeTwo}{\typeOne}) \lineararrow \typeTwo)}$, an initial accumulator $\valTwo$ of type $\typeTwo$, and folds $\valOne$ over a list $\valThree$ of type $\listt{}{\typeOne}$, to obtain a value of type $\typeTwo$. Folds are not sufficient to recover the full power of general recursion, but it appears that they are expressive enough to implement several quantum algorithms.


To conclude this section, we just remark how all of the \Quipper\ programs shown in Section \ref{sec:overview} can easily be encoded in \PQ. However, \PQ's system of simple types is not meant to reason about resources and it is therefore unable to tell us anything about the size of the circuits produced by these programs. Of course, the option of implementing \texttt{mapHadamard}, running it on a concrete input parameter and checking the size of the circuit produced at run-time is not ruled out, but this approach amounts to \emph{testing}, rather than \emph{verification}, and it is unable to yield general, parametric results on the size of the circuits produced by circuit families, i.e. algorithms.
\section{Circuits and Circuit Resource Metrics}
\label{sec:quantumcircuitmetrics}

So far, we have talked about circuits and their size metrics very informally, at a high level. This sections aims to formalize these two concepts, starting with what a \textit{quantum circuit} is, formally speaking, and proceeding to outline the class of metrics we are interested in analyzing.

\subsection{The Circuit Representation Language}

In the literature \cite{proto-quipper-m,proto-quipper-d,proto-quipper-l}, circuits are usually taken to be 
morphisms in some symmetric monoidal category. This model for circuits has the 
advantage of being abstract and highly compositional, but it is not 
particularly handy for reasoning about \textit{intensional} properties of 
circuits, such as their size. For this reason, we adopt instead the 
concrete circuit model presented in \cite{proto-quipper-r}, called 
\textit{Circuit Representation Language} (\crl) as the target of our notion of circuit description. As we said, individual circuits can be described in a gate-by-gate fashion, which is exactly what is done in \crl:
\begin{equation}
	\circSet \quad \circuitOne,\circuitTwo,\circuitThree,\circuitFour := \cidentity{\lcOne} \mid \circuitOne;\gateapp{\gateOne}{\bundle\labOne}{\bundle\labTwo},
\end{equation}
Here $\lcOne$ is a collection of typed wires, $\struct\labOne,\struct\labTwo$ are wire bundles as introduced in Section \ref{sec:the-proto-quipper-language}, and $\gateOne$ comes from an 
arbitrary, but fixed set $\gateSet$ of elementary operations. In other words, a 
circuit is either the identity on a collection of wires $\lcOne$ or a 
circuit followed by the application of operation $\gateOne$ to the wires in
$\bundle\labOne$. Operation $\gateOne$ outputs the wires in $\bundle\labTwo$. Note that we 
assume $\gateSet$ to include both quantum gates and more general operations 
such as measurements, wire initializations, etc. Also, circuits can be concatenated: we write $\concat{\circuitOne}{\circuitTwo}$ to denote $\circuitOne$, followed by all the elementary operations in $\circuitTwo$.

\subsection{Recursive Circuit Metrics}

Given a \crl\ circuit, how do we measure its size? That is, what is a \textit{circuit metric}? In general, we are interested in metrics that are defined recursively on the structure of a circuit, and this leads us to the following definition.

\begin{definition}[Recursive Circuit Metric]
	A \emph{recursive circuit metric} $\rcmOne$ is a pair
	$(\carrierOne,\rcmFunOne_\rcmOne)$,
	where $\carrierOne$ is a set called the \emph{carrier} of $\rcmOne$ and 
	$\rcmFunOne : \circSet \to \carrierOne$ is the \emph{metric function} 
	of $\rcmOne$, defined by structural recursion on circuits.
\end{definition}

In the wake of Section \ref*{sec:overview}, we distinguish between two fundamental 
classes of recursive circuit metrics: \textit{global} and \textit{local} 
metrics. Global metrics assign a single natural quantity to an entire circuit and in a way formalize the notion of \emph{size} of a circuit. 
Thus, we say that a recursive circuit metric is \textit{global} when 
$\carrierOne = \natSet$. Local metrics, on the other hand, assign a natural 
quantity to each \textit{wire} of a circuit. Thus, a local circuit metric has 
$\carrierOne = (\labelSet \to \natSet) \to (\labelSet \to \natSet)$, where 
$\labelSet$ is the set of wire labels. This signature means that a local 
metric is a higher-order function that, given an assignment of resource values 
to the \textit{input} wires of a circuit, returns an assignment of resource 
values to the \textit{output} wires of the circuit. The following examples 
illustrate how some common circuit metrics fall into this categorization.

\begin{example}[Global Metrics]
\label{ex:global-rcm}
The circuit metrics $\rcmGatecount = (\natSet, \gatecount)$ and $\rcmWidth = (\natSet, \width)$, where
	\begin{align*}
		\gatecount(\cidentity{\lcOne}) &= 0,\\
		\gatecount(\circuitOne;\gateapp{\gateOne}{\bundle\labOne}{\bundle\labTwo})
		 &= \gatecount(\circuitOne) + 1;
    \\\\
    \width(\cidentity{\lcOne}) &= |\lcOne|,\\
    \width(\circuitOne;\gateapp{\gateOne}{\bundle\labOne}{\bundle\labTwo}) &= \width(\circuitOne) + \max(0,(|\bundle\labTwo| - |\bundle\labOne|)-\discarded{\circuitOne}),
	\end{align*}
	are both global metrics. The quantity $(|\bundle\labTwo| - |\bundle\labOne|)$ denotes the number of new qubits or bits initialized by $\gateOne$, while $\discarded{\circuitOne}$ is the number of wires reusable from $\circuitOne$, obtained by subtracting the number of its outputs from its width.
\end{example}

\begin{example}[Local Metrics]
\label{ex:local-rcm}
	The circuit metric $\rcmDepth = ((\labelSet \to \natSet) \to \labelSet \to \natSet, \depth)$, where
	\begin{align*}
		\depth(\cidentity{\lcOne})(in)(\labThree) &= in(t), \\
		\depth(\circuitOne;\gateapp{\gateOne}{\bundle\labOne}{\bundle\labTwo})(in)(\labThree) &= \begin{cases}
			\max \{ \depth(\circuitOne)(in)(\labOne) \mid \labOne \in \freelabels(\bundle\labOne) \} + 1 & \textnormal{if } \labThree\in\freelabels(\bundle\labTwo),\\
			\depth(\circuitOne)(in)(\labThree) & \textnormal{otherwise},
		\end{cases}
	\end{align*}
  is a local metric. $\freelabels(\bundle\labOne)$ is the set of label names occurring in $\bundle\labOne$.
\end{example}

\subsection{Resource Aware Circuit Signatures}

The notions of global and local metrics that we have just given are able to effectively capture several common circuit resource metrics. Our end goal is to reason about these quantities within a type system, so the next step is deciding how to encode these definitions within a typing judgment. We proceed gradually, starting with a type system for circuits and moving up to circuit description languages only in the next section.
Circuits are traditionally amenable to a form of \textit{signature}, which tells us what are the inputs and outputs of a circuit. In the case of \crl, these signatures are derived via a rudimentary typing system which concludes judgments of the form $\circuitOne:\lcOne\to\lcTwo$, where $\lcOne$ and $\lcTwo$ are collections of typed wires.
Including global resource information into these judgments is not hard: it suffices to decorate them with a quantity corresponding to the value of the global metric on $\circuitOne$. This leads to a judgment of the form $\circuitOne:\lcOne\to\lcTwo;\indexOne$, where $\indexOne$ is what we previously called an \textit{index}, that is, an arithmetical expression representing a natural number.

Encoding local resources, on the other hand, is not as simple. This is because a local resource metric is not a natural value, but rather a function (and a higher-order function at that). We need our judgments to encode a function from assignments of \textit{something} to input labels to assignments of \textit{something} to output labels. Now, recall that we said that $\lcOne$ and $\lcTwo$ are ``collections of typed wires''. More formally, these are called \textit{label contexts}, and they are defined as mappings from (a subset of) label names to wire types. It therefore feels natural to rely on these mappings to encode local metrics. We need the assignment of local resource values to \textit{input} labels to be arbitrary, so we decorate the wire types in the codomain of $\lcOne$ with distinct \textit{index variables}. Let $\icontextTwo$ be the set of such variables. It is easy to see that by assigning different natural numbers to the variables in $\icontextTwo$, we find that $\lcOne$ encodes all the possible functions from the set of input label names to $\natSet$. On the other hand, we want the assignment of local resources to \textit{output} labels to depend on the resource values of the input labels, so we decorate the wire types in the codomain of $\lcTwo$ with index expressions built from the same variables in $\icontextTwo$. The precise form of these expressions depends of course on the shape of the circuit. This way, we can effectively encode a function $(\labelSet \to \natSet) \to (\labelSet \to \natSet)$ within a judgment, which now has the following form: 
\begin{equation}
  \icsJudgment{}{\icontextTwo}{\lcOne}{\circuitOne}{\lcTwo}{\indexOne}.
\end{equation}
We call this judgment a \textit{resource-aware circuit signature} (RACS). Before we move on the the derivation rules for these judgments, it is worth discussing exactly \emph{what kind of expressions} we use to annotate judgments and wire types. In previous sections, we said that indices are fundamentally arithmetic expressions, but is it enough to allow just the standard arithmetic operators within their language? On one hand, we saw in Section \ref{sec:overview} that a satisfactory resource analysis can be carried out using comparatively simple operations, such as maxima and additions. On the other hand, we also saw that different metrics require different operations to happen in different places. If we limited the language of indices to standard arithmetic operations, then we would have to provide \emph{distinct syntactical rules} for each metric, even at the basic level of the individual circuit. Because this is unwieldy, we follow a different path. Alongside the standard mathematical operations, we also allow a number of abstract resource operators to occur within indices. These operators represent how the local and global metrics evolve as we build up a \crl\ circuit, but they have no standard arithmetic interpretation by themselves. Rather, they act as placeholders, and are given different meanings depending on the metric under analysis.
How many such operators do we need? Ideally, we want to be able to replicate the recursive definition of circuit metrics as closely as possible. For this reason, we need an operator for each \crl\ constructor and for each metric kind (global and local). Knowing however that, due to its nature, the local metric of the identity circuit is always the identity function, we can make do with just three (families of) operators:
\begin{itemize}
  \item $\iid{\multiset\lcOne}$, which represents the global metric associated with the $\cidentity{\lcOne}$ circuit. Because we do not allow global metrics to depend on local metrics, $\iid{\multiset\lcOne}$ is indexed on the multiset of undecorated wire types in $\lcOne$'s codomain, denoted by $\multiset\lcOne$, as opposed to $\lcOne$ itself. For simplicity, we often write $\iid{\multiset\lcOne}$ as just $\multiset{\lcOne}$, with an abuse of notation.
  \item $\iappend{\gateOne}(\indexOne,\indexTwo,\indexThree,\indexFour)$, indexed on $\gateOne\in\gateSet$, which represents the global metric associated to the application of $\gateOne$ to a circuit of size $\indexOne$. $\indexTwo$ represents the size of the wires passing alongside $\gateOne$, while $\indexThree$ and $\indexFour$ are the size of the wires that go into and come out of $\gateOne$, respectively.
  \item $\igate{\gateOne}{n}(\indexOne_1,\dots,\indexOne_m)$, indexed on $\gateOne\in\gateSet$ and $n\in\natSet$, which represents the local metric value of the $n$-th output of $\gateOne$ as a function of the local resource values of its $m$ inputs.
\end{itemize}

This way, we can define a single set of syntactical rules for the derivation of RACSs, which is the one given in Figure \ref{fig:racs}. These rules in turn represent a \emph{family} of formalisms, indexed on the semantic interpretation of the abstract resource operators. We call such an interpretation a \emph{circuit metric interpretation} (CMI) and indicate it with the metavariable $\icsOne$ in subscript to the turnstile.

\begin{figure}
  \centering
  \fbox{
    \begin{mathpar}
      \inferrule[id]
      {\icsWFJudgment{\icontextTwo}{\lcOne}
        \\
        \iEqJudgment{\icsOne}{}{{\multiset\lcOne}}{\indexOne}}
      {\icsJudgment{\icsOne}{\icontextTwo}{\lcOne}{\cidentity{\lcOne}}{\lcOne}{\indexOne}}
      \and
      \inferrule[append]
      {\icsJudgment{\icsOne}{\icontextTwo}{\lcOne}{\circuitOne}{\lcTwo,\lcThree}{\indexOne}
        \\
        \iEqJudgment{\icsOne}{}{\iappend{\gateOne}{}(\indexOne,{\multiset\lcTwo},{\multiset\lcThree},{\multiset\lcFour})}{\indexTwo}
        \\\\
        \gatescheme{\gateOne} = (\icontextThree,\btypeOne,\btypeTwo)
        \\
        \isubThree \in \isemint{\icontextTwo}{\icontextThree}
        \\\\
        \bundleJudgment{\icsOne}{\icontextTwo}{\lcThree}{\bundle\labOne}{\isubThree(\btypeOne)}
        \\
        \bundleJudgment{\icsOne}{\icontextTwo}{\lcFour}{\bundle\labTwo}{\isubThree(\btypeTwo)}
      }
      {\icsJudgment{\icsOne}{\icontextTwo}{\lcOne}{\circuitOne;\gateapp{\gateOne}{\bundle\labOne}{\bundle\labTwo}}{\lcTwo,\lcFour}{\indexTwo}}
      \\
      \inferrule[unit]{}
      {\bundleJudgment{\icsOne}{\icontextTwo}{\emptycontext}{\unitv}{\unitt}}
      \and
      \inferrule[label]
      {\iEqJudgment{\icsOne}{\icontextTwo}{\indexOne}{\indexTwo}}
      {\bundleJudgment{\icsOne}{\icontextTwo}{\labOne:\wtypeOne^{\indexOne}}{\labOne}{\wtypeOne^{\indexTwo}}}
      \and
      \inferrule[tuple]
      {\bundleJudgment{\icsOne}{\icontextTwo}{\lcOne_1}{\bundle\labOne}{\btypeOne}
        \\
        \bundleJudgment{\icsOne}{\icontextTwo}{\lcOne_2}{\bundle\labTwo}{\btypeTwo}}
      {\bundleJudgment{\icsOne}{\icontextTwo}{\lcOne_1,\lcOne_2}{\tuple{\bundle\labOne}{\bundle\labTwo}}{\tensor{\btypeOne}{\btypeTwo}}}
    \end{mathpar}
  }
  \caption{Rules for RACSs and bundle judgments}
  \label{fig:racs}
\end{figure}

We will formally define what a CMI is in Section \ref{subsec:cmi}. For now, let us unpack the rules. First of all, note that $\lcOne,\lcTwo$ is shorthand for $\lcOne\uplus\lcTwo$, and denotes the union of two contexts with disjoint domains. $\bundleJudgment{\icsOne}{\icontextTwo}{\lcOne}{\struct\labOne}{\btypeOne}$ is what we call a \textit{bundle judgment}, and it is used to give a structured bundle type $\btypeOne$ to a label context $\lcOne$ via a wire bundle $\struct\labOne$. A label $\labOne$ has type $\wtypeOne^\indexOne$ when the local resource value of $\labOne$ is the function of the variables in $\icontextTwo$ described by $\indexOne$. On the other hand, $\iEqJudgment{\icsOne}{\icontextTwo}{\indexOne}{\indexTwo}$ is what we call a \textit{semantic judgment}, stating that for all assignments of natural numbers to the variables in $\icontextTwo$, indices $\indexOne$ and $\indexTwo$ are equal.
Moving on to the the actual derivation of RACSs, the \textsc{id} rule is pretty straightforward: as long as all of the index variables occurring in $\lcOne$ are in $\icontextTwo$ (written $\icsWFJudgment{\icontextTwo}{\lcOne}$), $\cidentity{\lcOne}$ has signature $\lcOne\to\lcOne$. Unsurprisingly, this encodes a local metric corresponding to the identity function, while the global metric is equal to ${\multiset{\lcOne}}$.
The \textsc{append} rule, on the other hand, is slightly harder to parse, and it describes the scenario shown in Figure \ref{fig:append}. Let us begin with the well typedness of the append: given an operation identifier $\gateOne$, we call $\gatescheme{\gateOne}$ the \textit{general type scheme} of $\gateOne$, defined as a triple $(\icontextThree,\btypeOne,\btypeTwo)$, where
\begin{itemize}
    \item $\icontextThree = \{\ivarOne_1,\dots,\ivarOne_n\}$ is a set of index variables,
		\item  $\btypeOne=\wtypeOne_1^{i_1} \otimes \wtypeOne_2^{i_2} \otimes \dots \otimes \wtypeOne_n^{i_n}$ is the input type of $\gateOne$,
		\item $\btypeTwo=\wtypeTwo_1^{\igate{\gateOne}{1}(i_1,i_2,\dots,i_n)} \otimes \wtypeTwo_2^{\igate{\gateOne}{2}(i_1,i_2,\dots,i_n)} \otimes \dots \otimes \wtypeTwo_m^{\igate{\gateOne}{m}(i_1,i_2,\dots,i_n)}$ is the output type of $\gateOne$.
\end{itemize}

Note that $\wtypeOne,\wtypeTwo\in\{\bitt,\qubitt\}$. In other words, a type scheme defines the base input and output types of an elementary operation, which are respectively decorated with fresh variables and abstract expressions describing the generic evolution of the local metric. In this regard, note that $\gatescheme{\gateOne}$ does \emph{not} depend on the current circuit metric interpretation. Rather, it is given alongside $\gateSet$.
A type scheme is instantiated to the appropriate wire annotations via an \textit{index substitution} $\isubThree$. Informally, $\isubThree\in\isemint{\icontextTwo}{\icontextThree}$ means that $\isubThree$ substitutes each of the index variables in $\icontextThree$ in an expression with an index whose free variables are in $\icontextTwo$. On the side of global resources, we have that the global metric value of $\circuitOne;\gateapp{\gateOne}{\bundle\labOne}{\bundle\labTwo}$ is given by $\iappend{\gateOne}{}(\indexOne,{\multiset{\lcTwo}},{\multiset{\lcThree}},{\multiset{\lcFour}})$.

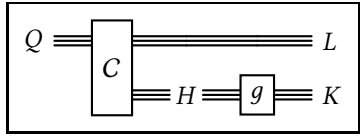
\begin{figure}
	\centering
	\fbox{\begin{quantikz}[classical gap=0.06cm, row sep = 5pt]
			\setwiretype{b} \lstick{$\lcOne$} & \gate[2]{\circuitOne} & & &  \rstick{$\lcTwo$}\\
			\setwiretype{n} & & \setwiretype{b} \ \push{\lcThree} \ & \gate{\gateOne} & \rstick{$\lcFour$}
	\end{quantikz}}
\caption{Routing of wires during the append of an operation $\gateOne$ to circuit $\circuitOne$}
\label{fig:append}
\end{figure}

\subsection{Circuit Metric Interpretations}
\label{subsec:cmi}

As we said earlier, the family of formalisms described by the rules in Figure \ref{fig:racs} is indexed on the semantic interpretation of the abstract resource operators, that is on a CMI $\icsOne$. Formally, $\icsOne$ is nothing more that a triple of (families of) functions $(\icsId{\icsOne}{\wtypeMultisetOne},\icsAppend{\icsOne}{\gateOne},\icsGate{\icsOne}{\gateOne}{\natOne})$, where $\icsId{\icsOne}{\multiset\lcOne} \in \natSet$ interprets $\multiset{\lcOne}$ (i.e. $\iid{\multiset\lcOne}$),
$\icsAppend{\icsOne}{\gateOne} : \natSet^4 \to \natSet$ interprets $\iappend{\gateOne}$, and 
$\icsGate{\icsOne}{\gateOne}{n} : \natSet^{\ginputs{\gateOne}}\to\natSet$ interprets $\igate{\gateOne}{n}$. In order to better understand how a CMI can capture a recursive circuit metric, we provide the following example CMIs for the familiar metrics of gate count, width and depth.

\begin{example}[Gate Count CMI]
  \label{ex:gatecount-cmi}
  In the case of gate count we have trivially
  \begin{align*}
    \icsId{\mathit{\icsGatecount}}{\wtypeMultisetOne} &= 0,\\
    \icsAppend{\mathit{\icsGatecount}}{\gateOne}(n,l,h,k) &= n+1.
  \end{align*}
  In this case, $\icsGate{\mathit{\icsGatecount}}{\gateOne}{n}$ is irrelevant and can be chosen to be any constant function. Note that this interpretation exactly encodes the recursive definition of gate count given earlier.
\end{example}

\begin{example}[Width CMI]
  \label{ex:width-cmi}
  In the case of width we have \begin{align*}
    \icsId{\icsWidth}{\wtypeMultisetOne} &= |\wtypeMultisetOne|,\\
    \icsAppend{\icsWidth}{\gateOne}(n,l,h,k) &= n + \max(0,k+l-n),
  \end{align*}
  while $\icsGate{\icsWidth}{\gateOne}{n}$ is once again irrelevant. Note that this interpretation also encodes the recursive definition of width given earlier, since $k+l-n = (k-h)-(n-l-h)$, where $k-h$ is exactly $|\struct\labTwo| - |\struct\labOne|$ and $n-l-h$ is equal to the width of the circuit minus the number of its outputs, i.e. to $\discarded{\circuitOne}$.
\end{example}

\begin{example}[Depth CMI]
  \label{ex:depth-cmi}
  For depth, we have
  \begin{equation*}
    \icsGate{\icsDepth}{\gateOne}{i}(n_1,\dots,n_k)=\max(n_1,\dots,n_k) + 1,
  \end{equation*}
  for all $\gateOne$ and $1 \leq i \leq \ginputs{\gateOne}$, while $\icsId{\icsDepth}{\wtypeMultisetOne}$ and $\icsAppend{\icsDepth}{\gateOne}$ are irrelevant and can be chosen to be any constant function. Although it is not as easy to check, we can see that this interpretation encodes the recursive definition of depth given earlier: for the wires that the $\gateOne$ is actually appended on, we get a depth of $\max(n_1,\dots,n_k) + 1$, which is equivalent to $\max \{ \depth(\circuitOne)(in)(\labOne) \mid \labOne \in \freelabels(\bundle\labOne) \} + 1$, whereas the depth of the other wires is left unchanged, i.e. stays at $\depth(\circuitOne)(in)(\labThree)$.
\end{example}

\subsection{Soundness of CMIs}

In the last three examples, we informally argued that the CMIs we gave capture the recursive definitions of the metrics they are associated with. We now give a more general definition of what it means for a CMI to be \emph{sound} with respect to a given recursive metric. The case of global metrics is straightforward: a CMI $\icsOne$ is sound with respect to a global metric $\rcmOne=(\natSet,\rcmFunOne_\rcmOne)$ when the derivability of a judgment of the form $\icsJudgment{\icsOne}{\icontextTwo}{\lcOne}{\circuitOne}{\lcTwo}{\indexOne}$ entails $\iLeqJudgment{\icsOne}{}{\rcmFunOne_\rcmOne(\circuitOne)}{\indexOne}$. Note that we do not require a CMI to be \emph{strictly equal} to the underlying recursive metric. At this stage, we are content with a sound overapproximation.
The case of local CMIs is of course more complicated. Recall that, given a local RCM $\rcmOne=((\labelSet\to\natSet)\to(\labelSet\to\natSet),\rcmFunOne_\rcmOne)$, we encode $\rcmFunOne_\rcmOne$ in a judgment $\icsJudgment{\icsOne}{\icontextTwo}{\lcOne}{\circuitOne}{\lcTwo}{\indexOne}$ by encoding the input of $\rcmFunOne_\rcmOne$ in $\lcOne$ and its output in $\lcTwo$. To do so, we annotate the wire types in $\lcOne$ with the index variables in $\icontextTwo$ and the wire types in $\lcTwo$ with index expressions that depend on the same variables. This leads to the following definition:

\begin{definition}[Locally Sound CMI]
	A CMI $\icsOne$ is \emph{sound with respect to a local metric $\rcmOne=((\labelSet\to\natSet)\to(\labelSet\to\natSet),\rcmFunOne_\rcmOne)$} when $\icsJudgment{\icsOne}{\icontextTwo}{(\labOne_1:\wtypeOne_1^{\ivarOne_1},\dots,\labOne_n:\wtypeOne_n^{\ivarOne_n})}{\circuitOne}{(\labTwo_1:\wtypeTwo_1^{\indexOne_1},\dots,\labTwo_m:\wtypeTwo_m^{\indexOne_m})}{\indexTwo}$ entails that for all $j\in\{1,\dots,m\}$, we have $\iLeqJudgment{\icsOne}{\icontextTwo}{\rcmFunOne_\rcmOne(\circuitOne)(f)(\labTwo_j)}{\indexOne_j}$, where $f(\labOne_h) = \ivarOne_p$ for $h\in\{1,\dots,n\}$.
\end{definition}

Because we have designed CMIs as to closely mimic the recursive definitions of circuit metrics, and because the formalism of Figure \ref{fig:racs} is small, it is not hard to prove that a given CMI is sound with respect to a given recursive metric. In fact, we prove that this is the case for the given examples. 

\begin{proposition}
  The $\icsGatecount$, $\icsWidth$ and $\icsDepth$ CMIs given in examples \ref{ex:gatecount-cmi}, \ref{ex:width-cmi} and \ref{ex:depth-cmi} are sound with respect to the corresponding recursive metrics given in examples \ref{ex:global-rcm} and \ref{ex:local-rcm}.
\end{proposition}
\begin{proof}
  By induction on the derivation of $\icsJudgment{\icsOne}{\icontextTwo}{\lcOne}{\circuitOne}{\lcTwo}{\indexOne}$, for $\icsOne\in\{\icsGatecount,\icsWidth,\icsDepth\}$.
\end{proof}
\section{From Circuits to Circuit Description Languages: Resource-Aware Proto-Quipper}
\label{sec:pqra}

In this section we take the step from the low-level \crl\ to an actual circuit description language. We present \emph{Resource Aware Proto-Quipper}, or \PQRA\, for short: an extension of the \PQ\ language introduced in Section \ref{sec:the-proto-quipper-language} equipped with a family of type-and-effect systems that supports generic reasoning about resources. Unsurprisingly, this step up in expressiveness introduces a number of new challenges with respect to the previous section. This is because although the object of analysis remains the same (quantum circuits), the structure of programs is fundamentally richer than those of circuits, allowing for greater parametricity and compositionality on one side, but requiring more sophisticated abstractions to keep track of resources on the other.

\subsection{Types and Terms}

\begin{figure}
	\centering
	\fbox{\begin{tabular}{l l r l}
			Types	& $\typeSet$	& $\typeOne,\typeTwo$ & $
			::= \unitt
			\mid \wtypeOne^{{\indexOne}}
			\mid \bang{\indexOne}{\typeOne}
			\mid \tensor{\typeOne}{\typeTwo}
			\mid \arrowt{\typeOne}{\typeTwo}{\indexOne}{{\btypeOne}}
			\mid \blistt{{\ivarOne}}{{\indexOne}}{\typeOne}
			\mid \mcirct{\indexOne}{{\icontextOne}}{\btypeOne}{\btypeTwo}$\\
			Param. types		& $\ptypeSet$	& $\ptypeOne,\ptypeTwo$ & $
			::= \unitt
			\mid \bang{\indexOne}{\typeOne}
			\mid \tensor{\ptypeOne}{\ptypeTwo}
			\mid \blistt{{\ivarOne}}{{\indexOne}}{\ptypeOne}
			\mid \mcirct{\indexOne}{{\icontextOne}}{\btypeOne}{\btypeTwo}$\\
			Bundle types	& $\btypeSet$	& $\btypeOne,\btypeTwo$ & $
			::= \unitt
			\mid \wtypeOne^{{\indexOne}}
			\mid \tensor{\btypeOne}{\btypeTwo}
			\mid \blistt{{\ivarOne}}{{\indexOne}}{\btypeOne}$\\
			&&&\\
			Terms	& $\termSet$	& $\termOne,\termTwo$ & $ ::= \app{\valOne}{\valTwo}
			\mid \dest{\varOne}{\varTwo}{\valOne}{\termOne}
			\mid \force{\valOne}
			\mid \mboxt{{\icontextOne}}{\btypeOne}{\valOne}$\\&&&$
			\mid \apply{\valOne}{\valTwo}
			\mid \return{\valOne}
			\mid \letin{\varOne}{\termOne}{\termTwo}
			\mid {\mfold{\ivarOne}{\valOne}{\valTwo}{\valThree}}$\\
			Values	& $\valSet$	& $\valOne,\valTwo$ & $
			::= \unitv
			\mid \varOne
			\mid \labOne
			\mid \abs{\varOne}{\typeOne}{\termOne}
			\mid \lift{\termOne}
			\mid \boxedCirc{\struct\labOne}{\circuitOne}{\struct{\labTwo}}
			\mid \tuple{\valOne}{\valTwo}$\\&&&$
			\mid \nil
			\mid {\rcons{\valOne}{\valTwo}}$\\
			Wire bundles \hspace{5pt}	& $\bvalSet$	& $\struct\labOne,\struct\labTwo$ & $
			::= \unitv
			\mid \labOne
			\mid \tuple{\struct\labOne}{\struct\labTwo}
			\mid \nil
			\mid {\rcons{\struct\labOne}{\struct\labTwo}}$\\	
			Indices & $\indexSet$ & $\indexOne,\indexTwo$ & $
			::= \natOne
			\mid \ivarOne
			\mid \iplus{\indexOne}{\indexTwo}
			\mid \iminus{\indexOne}{\indexTwo}
			\mid \imult{\indexOne}{\indexTwo}
			\mid \imax{\indexOne}{\indexTwo}$\\&&&$
			\mid \iid{\wtypeMultisetOne}
			\mid \iappend{\gateOne}(\indexOne,\indexTwo,\indexThree,\indexFour)
			\mid \igate{\gateOne}{n}(\indexOne_1,\dots,\indexOne_m)$\\&&&$
			\mid \imempty
			\mid \iwire{\wtypeOne}
			\mid \iseq{\indexOne}{\indexTwo}
			\mid \ipar{\indexOne}{\indexTwo}
			\mid \iseqBounded{\ivarOne}{\indexOne}{\indexTwo}
			\mid \iparBounded{\ivarOne}{\indexOne}{\indexTwo}
			$
	\end{tabular}}
	\caption{Types and syntax of \PQRA}
	\label{fig:pqra-syntax}
\end{figure}

The types and syntax of \PQRA\ are given in Figure \ref{fig:pqra-syntax}. On the side of types, we can see that the standard \PQ\ types have been \emph{refined} with quantitative information useful for resource analysis. Let us take a moment to dissect these refinements. Naturally, we inherit from circuit signatures the annotated wire type $\wtypeOne^\indexOne$, which encodes information about the local metric value of the corresponding label in $\indexOne$. Speaking of circuits, the boxed circuit type $\mcirct{\indexOne}{\icontextOne}{\btypeOne}{\btypeTwo}$ is annotated with $\indexOne$, which carries information about the global metric associated to the circuit. Local metric information is instead encoded in $\btypeOne$ and $\btypeTwo$, using locally-scoped index variables taken from $\icontextOne$. The linear arrow type $\arrowt{\typeOne}{\typeTwo}{\indexOne}{\btypeOne}$ is annotated with index $\indexOne$ and bundle type $\btypeOne$. Like we mentioned in Section \ref{sec:overview}, $\indexOne$ represents the global metric associated with the circuit built by the function once we apply it to a value of type $\typeOne$, while $\btypeOne$ is the type of the wires enclosed in the function's closure. The annotation on $\bang{\indexOne}{\typeOne}$ is similar to the arrow's and tells us that resuming (forcing) a suspended (lifted) computation of that type has the side effect of building a circuit of size at most $\indexOne$.
Lastly, the list type also gets annotated with quantitative information: the type $\blistt{\ivarOne}{\indexOne}{\typeOne}$ represents lists of \emph{exactly} $\indexOne$ elements, where the $\ivarOne$-th element has type $\typeOne$ which can depend on $\ivarOne$. We call this type a \emph{sized dependent list type}. Although this type might seem cumbersome to work with, size and dependency are essential for resource analysis: if we did not have size, then we would have no bound on how many wires are potentially in a list, whereas if we did not have dependency we would only be able to work with lists of wires with \emph{homogeneous} local metrics.

On the side of the language of terms, we have but minor changes from the \PQ\ definition of Section \ref{sec:the-proto-quipper-language}: $\mboxt{\icontextOne}{\btypeOne}{\valOne}$ now declares the set of local variables $\icontextOne$ used to encode local metrics in $\btypeOne$, and the $\foldoperator$ operator binds an index variable $\ivarOne$. This variable appears in the type of the step function and it allows each step of the fold to contribute \emph{differently} to the size of the constructed circuit. Lastly, note that lists now grow towards the right, with $\rcons{\valOne}{\valTwo}$ appending $\valTwo$ at the \emph{end} of list $\valOne$. This allows to work more smoothly with the new dependent aspect of lists.

\subsubsection{Index Terms}

In moving up to a higher level language, we keep the same approach to indices that we established in Section \ref{sec:quantumcircuitmetrics}, equipping their language with abstract operators that are interpreted differently depending on the resource being analyzed. The final language of indices, given in Figure \ref{fig:pqra-syntax}, comprises therefore of natural numbers, index variables and standard arithmetic operators (first row), the cicuit-level abstract operators introduced in section \ref{sec:quantumcircuitmetrics} (second row) and lastly a set of new abstract operators and constants that relate to the particular way that circuits are built in \PQ. In particular, individual circuits are no longer the focus of the resource analysis, but rather \emph{computations}. Annotations at the level of \PQRA\ therefore describe the size associated with the circuits produced by these computations. The new operators include:

\begin{itemize}
  \item $\imempty$, which describes a computation that does not build any circuit.
  \item $\iwire{\wtypeOne}$, which describes the size of a single wire of type $\wtypeOne$.
  \item $\iseq{\indexOne}{\indexTwo}$, which describes the size of the circuit built by the execution \emph{in sequence} of two computations annotated with $\indexOne$ and $\indexTwo$, respectively.
  \item $\ipar{\indexOne}{\indexTwo}$,  which describes the size of the circuit built by the execution \emph{in parallel} of two computations annotated with $\indexOne$ and $\indexTwo$, respectively.
\end{itemize}

Note that $\iseq{}{}$ and $\ipar{}{}$ also have their bounded versions: $\iseqBounded{\ivarOne}{\indexOne}{\indexTwo}$ and $\iparBounded{\ivarOne}{\indexOne}{\indexTwo}$, which describe the composition in sequence (respectivey, in parallel) of a family of $\indexOne$ computations, each annotated with global metric $\indexTwo$, where $\indexTwo$ may depend on $\ivarOne$.

\subsection{Typing Rules}

The idea of a type system in which computations are annotated with information regarding the side effect they produce (in this case, the size of the circuit they build) is not new and takes the name of a \emph{type-and-effect system} \cite{types-and-effects}. In the case of \PQRA, this means that a typing judgments on a term $\termOne$ gets annotated with an index $\indexOne$, denoting an upper bound to the size of the circuits produced by $\termOne$. We say \emph{circuits} because $\termOne$ can describe a \emph{family} of circuits, depending on some natural number parameters encoded as index variables. We are therefore dealing with a computational typing judgment of the form
\begin{equation}
  \itesCJudgment{}{\icontextOne}{\icontextTwo}{\contextOne}{\lcOne}{\termOne}{\typeOne}{\indexOne},
\end{equation}
which informally reads ``for all assignments of natural numbers to the index variables in $\icontextOne$ and $\icontextTwo$, under typing context $\contextOne$ and label context $\lcOne$, term $\termOne$ has type $\typeOne$ and produces a circuit of size at most $\indexOne$''. Here, $\contextOne$ is a traditional linear/nonlinear typing context, while $\lcOne$ is a label context like the ones we introduced in Section \ref{sec:quantumcircuitmetrics}. If $\contextOne$ only contains parameter variables, we write it as $\pcontextOne$. Note that we employ two different, disjoint index variable contexts, $\icontextOne$ and $\icontextTwo$. Why is that? That is because at the level of \PQRA, index variables play two fundamentally distinct roles.
On one hand, the index variables in $\icontextTwo$ are used to capture local metrics, in a way similar RACSs. These variables encode the assignment of local metrics to the outputs of some unspecified underlying circuit $\circuitOne$. $\lcOne$ represents the outputs of $\circuitOne$ that are manipulated by $\termOne$, and as such is annotated with index terms built \emph{exclusively} from the variables in $\icontextTwo$. The wire types occurring in $\typeOne$ are then annotated with indices that encode the assignment of local resource values to the outputs of the circuit built by $\termOne$, and thus may also depend on the variables in $\icontextTwo$.
On the other hand, the variables in $\icontextOne$ represent the natural number parameters that a computation may depend on, and they effectively allow a term $\termOne$ to build a \emph{family} of circuits, as opposed to an individual circuit. Contrary to the index variables in $\icontextTwo$, which are only used to syntactically encode a local metric function, the variables in $\icontextOne$ have an operational value, and as such have to be instantiated before $\termOne$ can be executed. Because different circuits in the same family may have different sizes \emph{and} different local metrics, both $\indexOne$ and the annotations of the wiretypes in $\typeOne$ may depend on the variables in $\icontextOne$.

Besides the main computational judgment, we also work with a judgment for values of the form $\itesVJudgment{}{\icontextOne}{\icontextTwo}{\contextOne}{\lcOne}{\valOne}{\typeOne}$. The latter is identical to the former, save for the absence of an effect annotation, since values are effect-less. That being said, the main typing rules for deriving these two kinds of judgment are given in Figure \ref{fig:ites-rules}, where $\typeOne \isub{\indexOne}{\ivarOne}$ denotes the capture-avoiding substitution of index $\indexOne$ for variable $\ivarOne$ in type $\typeOne$, $\itesIWFJudgment{\icontextOne}{\indexOne}$ means that all of the index variables occurring in $\indexOne$ are in $\icontextOne$, and $\itesTWFJudgment{\icontextOne}{\icontextTwo}{\typeOne}$ means that all of the index variables mentioned in type $\typeOne$ are in $\icontextOne$, with the exception of those occurring in the annotations of wire types, which may come from $\icontextTwo$.

\begin{figure}
	\centering
	\fbox{
		\begin{mathpar}[]
			\inferrule[var]
			{\itesTWFJudgment{\icontextOne}{\icontextTwo}{\pcontextOne,\varOne:\typeOne}}
			{\itesVJudgment{\itesOne,\icsOne}{\icontextOne}{\icontextTwo}{\pcontextOne,\varOne:\typeOne}{\emptycontext}{\varOne}{\typeOne}}
			\and
			\inferrule[lab]
			{\itesTWFJudgment{\icontextOne}{\icontextTwo}{\pcontextOne}
				\\
				\itesIWFJudgment{\icontextTwo}{\indexOne}}
			{\itesVJudgment{\itesOne,\icsOne}{\icontextOne}{\icontextTwo}{\pcontextOne}{\labOne:\wtypeOne^{\indexOne}}{\labOne}{\wtypeOne^{\indexOne}}}
			\\
			\inferrule[abs]
			{\itesCJudgment{\itesOne,\icsOne}{\icontextOne}{\icontextTwo}{\contextOne,\varOne:\typeOne}{\lcOne}{\termOne}{\typeTwo}{\indexOne}
				\\\\
				{\btypeOne=\abswirecontent{\contextOne}{\lcOne}{\termOne}}}
			{\itesVJudgment{\itesOne,\icsOne}{\icontextOne}{\icontextTwo}{\contextOne}{\lcOne}{\abs{\varOne}{\typeOne}{\termOne}}{\arrowt{\typeOne}{\typeTwo}{\indexOne}{\btypeOne}}}
			\and
			\inferrule[app]
			{\itesVJudgment{\itesOne,\icsOne}{\icontextOne}{\icontextTwo}{\pcontextOne,\contextOne_1}{\lcOne_1}{\valOne}{\arrowt{\typeOne}{\typeTwo}{\indexOne}{\btypeOne}}
				\\\\
				\itesVJudgment{\itesOne,\icsOne}{\icontextOne}{\icontextTwo}{\pcontextOne,\contextOne_2}{\lcOne_2}{\valTwo}{\typeOne}}
			{\itesCJudgment{\itesOne,\icsOne}{\icontextOne}{\icontextTwo}{\pcontextOne,\contextOne_1,\contextOne_2}{\lcOne_1,\lcOne_2}{\app{\valOne}{\valTwo}}{\typeTwo}{\indexOne}}
			\and
			\inferrule[circ]
			{\icsJudgment{\icsOne}{\icontextThree}{\lcOne}{\circuitOne}{\lcTwo}{\indexOne}
				\\
				\indexTwo = \imaxsug{\resourcecontent{\lcOne},\indexOne,\resourcecontent{\lcTwo}}
				\\\\
				\itesVJudgment{\itesOne,\icsOne}{\emptyset}{\icontextThree}{\emptycontext}{\lcOne}{\bundle\labOne}{\btypeOne}
				\\
				\itesVJudgment{\itesOne,\icsOne}{\emptyset}{\icontextThree}{\emptycontext}{\lcTwo}{\bundle\labTwo}{\btypeTwo}
        \\
				\itesTWFJudgment{\icontextOne}{\icontextTwo}{\pcontextOne}}
			{\itesVJudgment{\itesOne,\icsOne}{\icontextOne}{\icontextTwo}{\pcontextOne}{\emptycontext}{\boxedCirc{\bundle\labOne}{\circuitOne}{\bundle\labTwo}}{\mcirct{\indexTwo}{\icontextThree}{\btypeOne}{\btypeTwo}}}
			\and
			\inferrule[box]
			{\indexThree=\iseq{(\ipar{\indexOne}{\resourcecontent{\btypeOne}})}{\indexTwo}
      \\\\
      \itesVJudgment{\itesOne,\icsOne}{\icontextOne}{\icontextThree}{\pcontextOne}{\emptycontext}{\valOne}{\bang{\indexOne}{(\arrowt{\btypeOne}{\btypeTwo}{\indexTwo}{\btypeOne'})}}}
			{\itesCJudgment{\itesOne,\icsOne}{\icontextOne}{\icontextTwo}{\pcontextOne}{\emptycontext}{\mboxt{\icontextThree}{\btypeOne}{\valOne}}{\mcirct{\indexThree}{\icontextThree}{\btypeOne}{\btypeTwo}}{\imempty}}
			\and
			\inferrule[apply]
			{\itesVJudgment{\itesOne,\icsOne}{\icontextOne}{\icontextTwo}{\pcontextOne,}{\emptycontext}{\valOne}{\mcirct{\indexOne}{\icontextThree}{\btypeOne}{\btypeTwo}}
			\\\\
			\isubThree \in \isemint{\icontextTwo}{\icontextThree}
			\\
			\itesVJudgment{\itesOne,\icsOne}{\icontextOne}{\icontextTwo}{\pcontextOne,\contextOne}{\lcOne}{\valTwo}{\isubThree(\btypeOne)}}
			{\itesCJudgment{\itesOne,\icsOne}{\icontextOne}{\icontextTwo}{\pcontextOne,\contextOne}{\lcOne}{\apply{\valOne}{\valTwo}}{\isubThree(\btypeTwo)}{\indexOne}}
			\and
			\inferrule[lift]
			{\itesCJudgment{\itesOne,\icsOne}{\icontextOne}{\icontextTwo}{\pcontextOne}{\emptycontext}{\termOne}{\typeOne}{\indexOne}}
			{\itesVJudgment{\itesOne,\icsOne}{\icontextOne}{\icontextTwo}{\pcontextOne}{\emptycontext}{\lift{\termOne}}{\bang{\indexOne}{\typeOne}}}
			\and
			\inferrule[force]
			{\itesVJudgment{\itesOne,\icsOne}{\icontextOne}{\icontextTwo}{\pcontextOne}{\emptycontext}{\valOne}{\bang{\indexOne}{\typeOne}}}
			{\itesCJudgment{\itesOne,\icsOne}{\icontextOne}{\icontextTwo}{\pcontextOne}{\emptycontext}{\force{\valOne}}{\typeOne}{\indexOne}}
			\\
			\inferrule[return]
			{\itesVJudgment{\itesOne,\icsOne}{\icontextOne}{\icontextTwo}{\contextOne}{\lcOne}{\valOne}{\typeOne}}
			{\itesCJudgment{\itesOne,\icsOne}{\icontextOne}{\icontextTwo}{\contextOne}{\lcOne}{\return{\valOne}}{\typeOne}{\resourcecontent{\typeOne}}}
			\and
			\inferrule[let]
			{\indexThree = \iseq{(\ipar{\indexOne}{\resourcecontent{\contextOne_2,\lcOne_2}})}{\indexTwo}
      \\
      \itesCJudgment{\itesOne,\icsOne}{\icontextOne}{\icontextTwo}{\pcontextOne,\contextOne_1}{\lcOne_1}{\termOne}{\typeOne}{\indexOne}
				\\\\
				\itesCJudgment{\itesOne,\icsOne}{\icontextOne}{\icontextTwo}{\pcontextOne,\contextOne_2,\varOne:\typeOne}{\lcOne_2}{\termTwo}{\typeTwo}{\indexTwo}
        }
			{\itesCJudgment{\itesOne,\icsOne}{\icontextOne}{\icontextTwo}{\pcontextOne,\contextOne_1,\contextOne_2}{\lcOne_1,\lcOne_2}{\letin{\varOne}{\termOne}{\termTwo}}{\typeTwo}{\indexThree}}
			\\
      \inferrule[nil]
			{\itesTWFJudgment{\icontextOne}{\icontextTwo}{\pcontextOne}
				\\
				\itesTWFJudgment{\icontextOne,i}{\icontextTwo}{\typeOne}}
			{\itesVJudgment{\itesOne,\icsOne}{\icontextOne}{\icontextTwo}{\pcontextOne}{\emptycontext}{\nil}{\blistt{\ivarOne}{0}{\typeOne}}}
			\and
			\inferrule[rcons]
			{\itesVJudgment{\itesOne,\icsOne}{\icontextOne}{\icontextTwo}{\pcontextOne,\contextOne_1}{\lcOne_1}{\valOne}{\blistt{\ivarOne}{\indexOne}{\typeOne}}
				\\\\
				\itesVJudgment{\itesOne,\icsOne}{\icontextOne}{\icontextTwo}{\pcontextOne,\contextOne_2}{\lcOne_2}{\valTwo}{\typeOne\isub{\indexOne}{\ivarOne}}}
			{\itesVJudgment{\itesOne,\icsOne}{\icontextOne}{\icontextTwo}{\pcontextOne,\contextOne_1,\contextOne_2}{\lcOne_1,\lcOne_2}{\rcons{\valOne}{\valTwo}}{\blistt{\ivarOne}{\iplus{\indexOne}{1}}{\typeOne}}}
      \and
			\inferrule[fold]
			{\itesVJudgment{\itesOne,\icsOne}{\icontextOne,\ivarOne}{\icontextTwo}{\pcontextOne}{\emptycontext}{\valOne}{\bang{\indexOne}{(\arrowt{(\tensor{\typeTwo}{\typeOne\isub{\indexThree-(\ivarOne+1)}{\ivarTwo}})}{\typeTwo\isub{\iplus{\ivarOne}{1}}{\ivarOne}}{\indexTwo}{\btypeOne}})}
				\\
				\itesVJudgment{\itesOne,\icsOne}{\icontextOne}{\icontextTwo}{\pcontextOne,\contextOne_1}{\lcOne_1}{\valTwo}{\typeTwo\isub{0}{\ivarOne}}
				\\
				\itesVJudgment{\itesOne,\icsOne}{\icontextOne}{\icontextTwo}{\pcontextOne,\contextOne_2}{\lcOne_2}{\valThree}{\blistt{\ivarTwo}{\indexThree}{\typeOne}}
				\\
				\indexFour=\imaxsug{\resourcecontent{\typeTwo\isub{0}{\ivarOne}},\iseqBounded{\ivarOne}{\indexThree}{(\ipar{(\iseq{(\ipar{\ipar{\indexOne}{\resourcecontent{\typeTwo}}}{\resourcecontent{\typeOne\isub{\indexThree-(\ivarOne+1)}{\ivarTwo}}})}{\indexTwo})}{\iparBounded{\ivarTwo}{\indexThree-(\ivarOne+1)}{\resourcecontent{\typeOne}}})}}
			}
			{\itesCJudgment{\itesOne,\icsOne}{\icontextOne}{\icontextTwo}{\pcontextOne,\contextOne_1,\contextOne_2}{\lcOne_1,\lcOne_2}{\mfold{\ivarOne}{\valOne}{\valTwo}{\valThree}}{\typeTwo\isub{\indexThree}{\ivarOne}}{\indexFour}}
			\and
			\inferrule[vsub]
			{\itesVJudgment{\itesOne,\icsOne}{\icontextOne}{\icontextTwo}{\contextOne}{\lcOne}{\valOne}{\typeOne}
				\\
				\itesSubJudgment{\itesOne,\icsOne}{\icontextOne}{\icontextTwo}{\typeOne}{\typeTwo}}
			{\itesVJudgment{\itesOne,\icsOne}{\icontextOne}{\icontextTwo}{\contextOne}{\lcOne}{\valOne}{\typeTwo}}
			\and
			\inferrule[csub]
			{\itesCJudgment{\itesOne,\icsOne}{\icontextOne}{\icontextTwo}{\contextOne}{\lcOne}{\termOne}{\typeOne}{\indexOne}
				\\\\
				\itesSubJudgment{\itesOne,\icsOne}{\icontextOne}{\icontextTwo}{\typeOne}{\typeTwo}
				\\
				\iLeqJudgment{\itesOne}{\icontextOne}{\indexOne}{\indexTwo}}
			{\itesCJudgment{\itesOne,\icsOne}{\icontextOne}{\icontextTwo}{\contextOne}{\lcOne}{\termOne}{\typeTwo}{\indexTwo}}
	\end{mathpar}}
	\caption{Main \PQRA\ typing rules}
	\label{fig:ites-rules}
\end{figure}


\subsubsection{Size and Wire Content of Types}

Recall that wires have a size of their own, which at this level of abstraction is represented by the abstract index $\iwire{\wtypeOne}$.
Values may contain wires, which means that even though they do not \emph{actively} build circuits, they still have a size associated with them, corresponding to the aggregated size of their wires. This size has obviously an impact on global resource analysis, which is why it is imperative that we keep track of it in our type system. Fortunately, due to linearity, to each wire in a value $\valOne$ of type $\typeOne$ corresponds a wire type in $\typeOne$, which means that the size of a value can be computed as a function of its type.
We therefore define the following \emph{type size function} $\resourcecontent{\cdot} : \typeSet \to \indexSet$, which turns a type $\typeOne$ into an index that represents the size of the wires enclosed in a value of type $\typeOne$. We call $\resourcecontent{\typeOne}$ the \emph{size} of $\typeOne$.
\begin{mathpar}
	\resourcecontent{\unitt} = \resourcecontent{\bang\indexOne\typeOne} = \resourcecontent{\mcirct{\indexOne}{\icontextOne}{\btypeOne}{\btypeTwo}} = \imempty,
	\and
	\resourcecontent{\wtypeOne^{\indexOne}} = \iwire{\wtypeOne},
	\and
	\resourcecontent{\tensor{\typeOne}{\typeTwo}} = \ipar{\resourcecontent{\typeOne}}{\resourcecontent{\typeTwo}},
	\and
	\resourcecontent{\blistt{\ivarOne}{\indexOne}{\typeOne}} = \iparBounded{\ivarOne}{\indexOne}{\resourcecontent\typeOne},
	\and
	\resourcecontent{\arrowt{\typeOne}{\typeTwo}{\indexOne}{\btypeOne}} = \resourcecontent{\btypeOne}.
\end{mathpar}
Typing and label contexts can also contain wires, so we lift this definition to contexts in the natural way: $\resourcecontent{\contextOne;\lcOne}$ is thus the combined size of the wires occurring in $\contextOne$ and $\lcOne$.
The type size function plays a fundamental role in many rules. In \textsc{return}, it gives us the very size of the trivial circuit built by computation $\return{\valOne}$, i.e. the identity circuit on the wires in $\valOne$. In \textsc{let}, the wires in $\contextOne_2$ and $\lcOne_2$ that are used to type $\termTwo$ have to flow alongside the circuit of size $\indexOne$ built by $\termOne$, in parallel, and as such their size is taken into account in $\ipar{\indexOne}{\resourcecontent{\contextOne_2;\lcOne_2}}$. We encounter a similar scenario in the \textsc{box} rule, where the inputs to the lifted circuit-building function $\valOne$, which have size $\resourcecontent{\btypeOne}$, have to be routed around the circuit of size $\indexOne$ built when forcing $\valOne$ itself (see Section \ref{sec:pqr-operational-semantics}, Figure \ref{fig:operational-semantics}). The \textsc{fold} rule also makes extensive use of the type size function, as during the $i$-th iteration of the fold the wires in the remaining elements of the list, whose size amounts to $\iparBounded{\ivarTwo}{\indexThree-(\ivarOne+1)}{\resourcecontent{\typeOne}}$, have to flow alongside the circuit built by the step function, which has size $\iseq{(\ipar{\ipar{\indexOne}{\resourcecontent{\typeTwo}}}{\resourcecontent{\typeOne\isub{\indexThree-(\ivarOne+1)}{\ivarTwo}}})}{\indexTwo}$ due to the same reasons as \textsc{box}.

Overall, the definition of $\resourcecontent{\cdot}$ is fairly straightforward, save for the arrow type case. If we were working with regular arrow types, we would not be able to know how many wires are captured by a function from its environment, as the types of these wires are not guaranteed to occur in either the domain or the codomain of the arrow. This explains why we have to explicitly annotate the arrow type with this piece of information, as we anticipated in previous sections. 
But where is annotation $\btypeOne$ computed? Looking at the \textsc{abs} rule, we can see $\btypeOne=\abswirecontent{\contextOne}{\lcOne}{\termOne}$. Here, $\wirecontent{\cdot}$ is what we call the \emph{wire content function}, which is originally defined on types. What $\wirecontent{\typeOne}$ does is ``strip'' $\typeOne$ of anything but wire types, obtaining  a bundle type that represents all and only the wires that are contained in a value of type $\typeOne$:
\begin{mathpar}
	\wirecontent{\unitt} = \wirecontent{\bang\indexOne\typeOne} = \wirecontent{\mcirct{\indexOne}{\icontextOne}{\btypeOne}{\btypeTwo}} = \unitt,
	\and
  \wirecontent{\wtypeOne^{\indexOne}} = \wtypeOne^{\indexOne},
	\and
	\wirecontent{\tensor{\typeOne}{\typeTwo}} = \tensor{\wirecontent\typeOne}{\wirecontent\typeTwo},
	\and
	\wirecontent{\blistt{\ivarOne}{\indexOne}{\typeOne}} = \blistt{\ivarOne}{\indexOne}{\wirecontent{\typeOne}},
	\and
	\wirecontent{\arrowt{\typeOne}{\typeTwo}{\indexOne}{\btypeOne}} = \btypeOne.
\end{mathpar}

The $\wirecontent{\cdot}$ function is similar to $\resourcecontent{\cdot}$, but it preserves more information. In fact, $\resourcecontent{\wirecontent{\typeOne}} = \resourcecontent{\typeOne}$ for all $\typeOne$. In lifting $\wirecontent{\cdot}$ to contexts, instead of assuming a total ordering of label and variable names, we use the structure of $\termOne$ as a guide to determine the shape of the resulting bundle type. This may seem like an unnecessarily complicated step, but it is a good way of ensuring that the structure of $\btypeOne$ is preserved under variable substitution, something that is essential for the correctness proof of Section \ref{sec:correctness} and for a future denotational account of \PQRA.

\subsubsection{Boxed Circuits} Notice how the typing judgments in Figure \ref{fig:ites-rules} are parametrized on a CMI $\icsOne$. We use $\icsOne$ when typing boxed circuits, which we recall are little more than language-level wrappers for circuit objects. As a result, $\icsOne$ provides a sort of ``ground truth'' regarding the size and local metrics of the elementary gates and operations encoded as boxed circuits in \PQRA.
This happens in the \textsc{circ} rule: when circuit $\circuitOne$ has signature $\icsJudgment{\icsOne}{\icontextThree}{\lcOne}{\circuitOne}{\lcTwo}{\indexOne}$ and $\struct\labOne$ and $\struct\labTwo$ confer types $\btypeOne$ and $\btypeTwo$ to $\lcOne$ and $\lcTwo$, respectively, the boxed circuit $\boxedCirc{\bundle\labOne}{\circuitOne}{\bundle\labTwo}$ has type $\mcirct{\indexTwo}{\icontextThree}{\btypeOne}{\btypeTwo}$, where $\indexTwo = \imaxsug{\resourcecontent{\lcOne},\indexOne,\resourcecontent{\lcTwo}}$ for correctness purposes. A boxed circuit type therefore encodes almost all of the information that we can extract from circuit signatures using $\icsOne$ as the metric of reference.

That being said, it comes as no surprise that the \textsc{apply} rule is reminiscent of the \textsc{append} rule from Figure \ref{fig:racs}, especially with respect to local resources: the circuit type $\mcirct{\indexOne}{\icontextThree}{\btypeOne}{\btypeTwo}$ effectively acts as a general type scheme $(\icontextThree,\btypeOne,\btypeTwo)$, and the substitution $\isubThree\in\isemint{\icontextTwo}{\icontextThree}$ instantiates the local index variables in $\icontextThree$ with the actual variable names used to encode metrics in the current type derivation. 

\subsubsection{Subtyping}
Being based in part on type refinements, our type system unsurprisingly features two subsumption rules: \textsc{vsub} and \textsc{csub}. These rules rely on a subtyping judgment of the form $\itesSubJudgment{}{\icontextOne}{\icontextTwo}{\typeOne}{\typeTwo}$, which informally reads ``for all assignments of natural numbers to the index variables in $\icontextOne$ and $\icontextTwo$, $\typeOne$ is at least as refined as $\typeTwo$''.
In the context of resources analysis, this usually means that $\typeOne$ describes a resource consumption upper bound no greater than the one described by $\typeTwo$. More specifically, on the side of local resources,  $\wtypeOne^\indexOne$ is a subtype of $\wtypeOne^\indexTwo$ when $\indexOne$ is no greater than $\indexTwo$. On the side of global resources, we require the resource annotations on the $\bang{}{}$, $\arrowt{}{}{}{}$ and $\mathsf{Circ}$ constructors to be smaller in the subtype than they are in the supertype. Note that list types \emph{of equal length} are covariant in their type parameter.

\subsection{Resource Metric Interpretations}

Besides the CMI $\icsOne$, which provides the ground truth in the matter of the size of circuits, the typing judgments in Figure \ref{fig:ites-rules} are also parametrized on a second interpretation $\itesOne$, which we call a \emph{resource metric interpretation} (RMI) and which specifically interprets the abstract operators introduced with \PQRA. An RMI is therefore a quadruple $(\itesMempty{\itesOne},\itesWires{\itesOne}{\wtypeOne},\itesSeq{\itesOne},\itesPar{\itesOne})$, where $\itesMempty{\itesOne}\in\natSet$ interprets $\imempty$, $\itesWires{\itesOne}{\wtypeOne} \in \natSet$ interprets $\iwire{\wtypeOne}$, $\itesSeq{\itesOne} : \natSet \times \natSet \to \natSet$ interprets $\iseq{}{}$ and its bounded counterpart and $\itesPar{\itesOne} : \natSet \times \natSet \to \natSet$ interprets $\ipar{}{}$ and its bounded counterpart. Recall when in Section \ref{sec:overview} we discussed informally the size of wires and how different global metrics are combined in sequence and parallel. In the following examples, we can finally give a formal account of those intuitions.

\begin{example}[Gate Count RMI]
  \label{ex:gatecount-rmi}
  In the case of gate count we have 
  \begin{align*}
  \itesMempty{\itesGatecount} = \itesWires{\itesGatecount}{\wtypeOne}&=0,\\
  \itesSeq{\itesGatecount} = \itesPar{\itesGatecount} &= +.
  \end{align*}
\end{example}

\begin{example}[Width RMI]
  \label{ex:width-rmi}
  In the case of width we have \begin{align*}
    \itesMempty{\itesWidth}&=0, &
    \itesWires{\itesWidth}{\wtypeOne}&=1,\\
    \itesSeq{\itesWidth} &= \max, &
    \itesPar{\itesWidth} &= +.
  \end{align*}
\end{example}

The attentive reader might have noticed that RMIs do not handle local resources. This is because they do not need to: in moving from \crl\ to \PQRA, the way we handle circuits (and therefore \emph{global} metrics) changes, but the way we handle wires (and therefore \emph{local} metrics) remains fundamentally unchanged. If the inputs and outputs of the circuits that we box are soundly annotated with local resource metrics, then by instantiating the index variables occurring in them like we do in the \textsc{apply} rule we effectively propagate the local resource analysis defined in Section \ref{sec:quantumcircuitmetrics} to the circuit description language, at no extra cost.

At this point, a question arises naturally: should we consider any and all possible interpretations of the aforementioned resource operators? Or can we safely make some assumptions about the properties of the RMIs we work with, as to simplify reasoning? The answer is clearly that we \emph{do} want to impose some constraints on RMIs. We say that an RMI $\itesOne$ is \emph{well-behaved} when it satisfies the following properties:
\begin{enumerate}
  \item $(\natSet,\itesPar{\itesOne})$ is an ordered commutative monoid with identity $\itesMempty{\itesOne}$,
  \item $(\natSet,\itesSeq{\itesOne})$ is an ordered monoid with identity $\itesMempty{\itesOne}$,
  \item $\itesPar{\itesOne}(\itesWires{\itesOne}{\wtypeOne},\itesSeq{\itesOne}(n,m)) = \itesSeq{\itesOne}(\itesPar{\itesOne}(\itesWires{\itesOne}{\wtypeOne},n),\itesPar{\itesOne}(\itesWires{\itesOne}{\wtypeOne},m))$ for all $\wtypeOne\in\wtypeSet,n,m\in\natSet$.
  \item Either $\forall \wtypeOne.\itesWires{\itesOne}{\wtypeOne} = \itesMempty{\itesOne}$ or $\forall n,m . n\leq m \implies \itesSeq{\itesOne}(n,m) = \itesSeq{\itesOne}(m,n) = m$.
\end{enumerate}

Intuitively, well-behavedness describes what we expect from a ``reasonable'' resource analysis. In particular, the requirements that the monoids in constraints (1) and (2) be ordered reflects the idea that when we compose a circuit $\circuitOne$ with another circuit $\circuitTwo$, the result is \emph{at least as big} as either $\circuitOne$ or $\circuitTwo$, i.e. $\itesSeq{\itesOne}$ and $\itesPar{\itesOne}$ should be monotonically non-decreasing. The additional commutativity of $\itesPar{\itesOne}$ reflects the idea that that computations executing in parallel are independent, and thus their relative ordering is irrelevant for resource analysis. Constraint (3) is more subtle, and tells us that $\itesPar{\itesOne}$ distributes over $\itesSeq{\itesOne}$ when wires are involved. Finally, constraint (4) tells us that any metric that takes wires into consideration (e.g. width) should also take into account their reuse during sequential composition. Note that these properties are only meant to simplify the formal analysis of resources, and they \emph{not} guarantee the correctness of \PQRA\ (see Section \ref{sec:correctness}).
Naturally, the RMIs given in examples \ref{ex:gatecount-rmi} and \ref{ex:width-rmi} are easily proven to be well-behaved. In fact, we assume that from now on all the RMIs we deal with are well-behaved, unless otherwise stated.


\subsection{Operational Semantics}
\label{sec:pqr-operational-semantics}

We define the operational semantics of \PQRA\ by means of a big-step evaluation relation on \emph{configurations}. A configuration consists of a circuit $\circuitOne$ and a term $\termOne$ that can access $\circuitOne$'s output wires and thus extend the circuit with new operations. We write $\config{\circuitOne}{\termOne}\eval\config{\circuitTwo}{\valOne}$ when $\termOne$ evaluates to $\valOne$ and turns $\circuitOne$ into $\circuitTwo$ as a side effect. The rules used to derive this evaluation relation are shown in Figure \ref{fig:operational-semantics}.

\begin{figure}
	\centering
	\resizebox{\linewidth}{!}{
		\fbox{\begin{mathpar}
      \mprset {sep=1em}
			\inferrule[{app}]
			{\config{\circuitOne}{\termOne\sub{\valOne}{\varOne}}
				\eval \config{\circuitTwo}{\valTwo}}
			{\config{\circuitOne}{\app{(\abs{\varOne}{\typeOne}{\termOne})}{\valOne}}
				\eval \config{\circuitTwo}{\valTwo}}
			\and
			\inferrule[{dest}]
			{\config{\circuitOne}{\termOne\sub{\valOne}{\varOne}\sub{\valTwo}{\varTwo}}
				\eval \config{\circuitTwo}{\valThree}}
			{\config{\circuitOne}{\dest{\varOne}{\varTwo}{\tuple{\valOne}{\valTwo}}{\termOne}}
				\eval \config{\circuitTwo}{\valThree}}
			\and
			\inferrule[{force}]
			{\config{\circuitOne}{\termOne} \eval \config{\circuitTwo}{\valOne}}
			{\config{\circuitOne}{\force{(\lift{\termOne})}} \eval \config{\circuitTwo}{\valOne}}
			\and
			\inferrule[{apply}]
			{\config{\circuitThree}{\struct{\labFour}} = \append{\circuitOne}{\struct{\labThree}}{\boxedCirc{\struct\labOne}{\circuitTwo}{\struct{\labTwo}}}}
			{\config{\circuitOne}{\apply{\boxedCirc{\struct\labOne}{\circuitTwo}{\struct{\labTwo}}}{\struct{\labThree}}} \eval \config{\circuitThree}{\struct{\labFour}}}
			\and
			\inferrule[{box}]
			{(\lcOne,\struct{\labOne})=\freshlabels{\mtypeOne}
				\\
				\config{\cidentity{\lcOne}}{\termOne} \eval \config{\circuitTwo}{\valOne}
				\\
				\config{\circuitTwo}{\app{\valOne}{\struct{\labOne}}} \eval \config{\circuitThree}{\struct{\labTwo}}}
			{\config{\circuitOne}{\mboxt{\icontextThree}{\btypeOne}{(\lift{\termOne})}} \eval \config{\circuitOne}{\boxedCirc{\struct{\labOne}}{\circuitThree}{\struct{\labTwo}}}}
			\\
			\inferrule[{return}]
			{ }
			{\config{\circuitOne}{\return{\valOne}} \eval \config{\circuitOne}{\valOne}}
			\and
			\inferrule[{fold-end}]
			{ }
			{\config{\circuitOne}{\mfold{\ivarOne}{\valOne}{\valTwo}{\nil}} \eval \config{\circuitOne}{\valTwo}}
			\\
			\inferrule[{let}]
			{\config{\circuitOne}{\termOne} \eval \config{\circuitTwo}{\valOne}
				\\\\
				\config{\circuitTwo}{\termTwo\sub{\valOne}{\varOne}} \eval \config{\circuitThree}{\valTwo}}
			{\config{\circuitOne}{\letin{\varOne}{\termOne}{\termTwo}} \eval \config{\circuitThree}{\valTwo}}
			\and
			\inferrule[{fold-step}]
			{\config{\circuitOne}{\termOne\isub{0}{\ivarOne}} \eval \config{\circuitTwo}{\valFour}
				\\
				\config{\circuitTwo}{\app{\valFour}{\tuple{\valOne}{\valThree}}} \eval \config{\circuitThree}{\valFive}
				\\\\
				\config{\circuitThree}{\mfold{\ivarOne}{(\lift{\termOne\isub{\iplus{\ivarOne}{1}}{\ivarOne}})}{\valFive}{\valTwo}} \eval \config{\circuitFour}{\valSix}}
			{\config{\circuitOne}{\mfold{\ivarOne}{(\lift{\termOne})}{\valOne}{(\rcons{\valTwo}{\valThree})}} \eval \config{\circuitFour}{\valSix}}
	\end{mathpar}}
	}
	\caption{\PQRA's big-step operational semantics}
	\label{fig:operational-semantics}
\end{figure}

The \textsc{apply} rule describes how $\applyoperator$ extends the underlying circuit $\circuitOne$ with subcircuit $\circuitTwo$. The concatenation of $\circuitTwo$ to $\circuitOne$ is delegated to the $\appendfunction$ function, which is given in Definition \ref{def:append}. We write $\boxedCirc{\struct\labOne}{\circuitOne}{\struct\labTwo} \cong \boxedCirc{\struct\labThree}{\circuitTwo}{\struct\labFour}$ to say that two boxed circuits are \emph{equivalent}, i.e. that they are the same circuit up to a renaming of labels. Note that this amounts to a form of $\alpha$-equivalence for circuits.

\begin{definition}[$\appendfunction$]
	\label{def:append}
	We define \emph{the emission of $\boxedCirc{\struct\labOne}{\circuitTwo}{\struct{\labTwo}}$ to $\circuitOne$ on $\struct\labThree$}, written $\append{\circuitOne}{\struct\labThree}{\boxedCirc{\struct\labOne}{\circuitTwo}{\struct{\labTwo}}}$, as a pair of circuit and wire bundle computed as follows:
	\begin{enumerate}
		\item Find $\boxedCirc{\struct\labThree}{\circuitTwo'}{\struct{\labFour}} \cong \boxedCirc{\struct\labOne}{\circuitTwo}{\struct{\labTwo}}$ such that the labels shared by $\circuitOne$ and $\circuitTwo'$ are exactly those in $\struct\labThree$,
		\item Return $\config{\concat{\circuitOne}{\circuitTwo'}}{\struct{\labFour}}$.
	\end{enumerate}
\end{definition}

On the other hand, the semantics of a term of the form $\mboxt{\icontextThree}{\mtypeOne}{(\lift{\termOne})}$ relies on the $\freshlabelsfunction$ function, which takes as input a bundle type $\mtypeOne$ and instantiate fresh $\lcOne,\struct\labOne$ such that $\struct\labOne$ has type $\btypeOne$ under $\lcOne$. The wire bundle $\struct\labOne$ is then passed as an argument to the circuit-building function $\valOne$, and the resulting computation is evaluated in the context of the identity circuit $\cidentity{\lcOne}$. The result is a standalone circuit $\circuitThree$, together with its output labels $\struct{\labTwo}$. Eventually, $\struct{\labOne}$ and $\struct{\labTwo}$ become respectively the input and output interfaces of the boxed circuit $\boxedCirc{\struct{\labOne}}{\circuitThree}{\struct{\labTwo}}$, which is the result of the evaluation.

Note, at this point, that $\mtypeOne$ controls how many labels are initialized by the $\freshlabelsfunction$ function. For instance, if $\mtypeOne$ was equal to $\blistt{i}{4}{\qubitt}$, we would need to initialize $4$ new qubits. It follows that \PQRA\ indices are not only relevant to typing, but they have operational significance as well. For this reason, the aforementioned semantics are well-defined only on terms closed both in the sense of regular variables \textit{and} index variables.
The \textsc{fold-step} rule is especially interesting in this aspect. Earlier we said that $\foldoperator$ binds a variable name $\ivarOne$, which allows each iteration of the fold to contribute differently to the overall circuit being built. This mechanism is evident in this rule: when we run the first iteration, $\ivarOne$ is instantiated to $0$ and the resulting step function is applied to the accumulator. Next, the result of the iteration becomes the new accumulator, we increment $\ivarOne$ (effectively obtaining a different step function) and we recur on the rest of the list. Note that $\valOne,\valTwo,\valThree,\valFour,\valFive$ and $\valSix$ are all values.

\section{Correctness}
\label{sec:correctness}

In this section we define what it means for a \PQRA\ judgment to be \emph{correct}, and we prove that a \PQRA\ instance defined by an RMI $\itesOne$ and a CMI $\icsOne$ is correct by construction under reasonable assumptions about the relationship between $\itesOne$ and $\icsOne$.

\subsection{Defining Correctness}

Any property of a \PQRA\ judgment fundamentally depends on the chosen RMI and CMI: the former describes the high-level analysis carried out at the language level, whilst the second describes (an overapproximation of) the actual size of the circuits. For this reason, from now on, when we talk of correctness, we are really talking about the correctness \emph{of} an RMI \emph{with respect} to a CMI. If an RMI gives us an instance of \PQRA\ that is correct with respect to a CMI, and that CMI soundly overapproximates a ground truth recursive circuit metric, then we can rightfully say that the analyses carried out through the relevant \PQRA\ instance are correct.

Secondly, we must distinguish between the correctness of the analysis of global and local resources.
On the side of global metrics, correctness is fairly straightforward: a judgment $\itesCJudgment{\itesOne,\icsOne}{\emptyset}{\icontextTwo}{\emptycontext}{\lcOne}{\termOne}{\typeOne}{\indexOne}$ is correct when, for all $\circuitOne$ supplying the wires in $\lcOne$, $\termOne$ evaluates to $\valOne$, extending $\circuitOne$ with a subcircuit $\circuitTwo$ whose size (according to $\icsOne$) is at most $\indexOne$.
On the side of refinement types and local metrics, correctness is a bit more involved: let $\lcTwo$ be the label context corresponding to the labels output by $\circuitTwo$. Due to linearity, the labels in the domain of $\lcTwo$ must appear \emph{somewhere} in $\valOne$. Similarly, the wire types in the codomain of $\lcTwo$, alongside their local metric annotations, must appear \emph{somewhere} in $\typeOne$. Thus, correctness in the sense of local resources is a matter of matching the annotations in $\lcTwo$ with the annotations in $\typeOne$ using $\valOne$ as a bridge, and then checking that the annotations in $\typeOne$ soundly overapproximate the corresponding annotations in $\lcTwo$. This idea is formalized in the following definition of \emph{local resource correctness}.

\begin{definition}[Local Resource Correctness]
	We define $\vcorrect{\itesOne,\icsOne}{\icontextTwo}{\lcOne}{\valOne}{\typeOne}$ as the smallest relation such that the following conditions hold:
  \begin{align*}
		\vcorrect{\itesOne,\icsOne}{\icontextTwo}{\emptycontext}{\valOne}{\unitt},
    \qquad
    \vcorrect{\itesOne,\icsOne}{\icontextTwo}{\emptycontext}{\valOne}{\bang{\indexOne}{\typeOne}},&
    \qquad
    \vcorrect{\itesOne,\icsOne}{\icontextTwo}{\emptycontext}{\valOne}{\mcirct{\indexOne}{\icontextThree}{\btypeOne}{\btypeTwo}},
		\\
		\vcorrect{\itesOne,\icsOne}{\icontextTwo}{\emptycontext}{\nil}{\blistt{\ivarOne}{\indexOne}{\typeOne}} \iff &\; \iEqJudgment{}{}{\indexOne}{0},
		\\
		\vcorrect{\itesOne,\icsOne}{\icontextTwo}{\labOne:\wtypeOne^{\indexOne}}{\labOne}{\wtypeOne^{\indexTwo}} \iff &\; \iLeqJudgment{\itesOne,\icsOne}{\icontextTwo}{\indexOne}{\indexTwo},
		\\
		\vcorrect{\itesOne,\icsOne}{\icontextTwo}{\lcOne}{\abs{\varOne}{\typeOne}{\termOne}}{\arrowt{\typeOne}{\typeTwo}{\indexOne}{\btypeOne}} \iff &\; \vcorrect{\itesOne,\icsOne}{\icontextTwo}{\lcOne}{\labels(\termOne)}{\btypeOne},
		\\
		\vcorrect{\itesOne,\icsOne}{\icontextTwo}{\lcOne\uplus\lcTwo}{\tuple{\valOne}{\valTwo}}{\tensor{\typeOne}{\typeTwo}} \iff &\; \vcorrect{\itesOne,\icsOne}{\icontextTwo}{\lcOne}{\valOne}{\typeOne} \wedge \vcorrect{\itesOne,\icsOne}{\icontextTwo}{\lcTwo}{\valTwo}{\typeTwo},
		\\
		\vcorrect{\itesOne,\icsOne}{\icontextTwo}{\lcOne\uplus\lcTwo}{\rcons{\valOne}{\valTwo}}{\blistt{\ivarOne}{\indexOne}{\typeOne}} \iff &\; \exists \indexTwo.(\iEqJudgment{}{}{\indexOne}{\indexTwo+1} \wedge \vcorrect{\itesOne,\icsOne}{\icontextTwo}{\lcOne}{\valOne}{\blistt{\ivarOne}{\indexTwo}{\typeOne}}\\
      &\wedge \vcorrect{\itesOne,\icsOne}{\icontextTwo}{\lcTwo}{\valTwo}{\typeOne\isub{\indexTwo}{\ivarOne}})
	\end{align*}
\end{definition}
Here $\labels : \termSet \cup \valSet \to \bvalSet$ is a function that strips a language expression of everything that is not a wire bundle, in the same way that $\wirecontent{\cdot}$ strips a type of anything that is not a wire type. The resulting wire bundle encodes all and only the occurrences of labels within the expression's syntax tree, which have to be correct with respect to $\btypeOne$. Now that we have a formal definition, we can define what it means for a \PQRA\ judgment, and by extension for an RMI, to be correct.


\begin{definition}[Correct Judgment]
	We say that a judgment $\itesCJudgment{\itesOne,\icsOne}{\emptyset}{\icontextTwo}{\emptycontext}{\lcOne}{\termOne}{\typeOne}{\indexOne}$ is \emph{correct} when for all $\circuitOne$ such that $\icsJudgment{\icsOne}{\icontextTwo}{\lcTwo}{\circuitOne}{\lcOne,\lcThree}{\indexTwo}$ there exist $\circuitTwo,\valOne,\lcFour,\indexThree$ such that
  \begin{equation*}
		\config{\circuitOne}{\termOne} \eval \config{\concat{\circuitOne}{\circuitTwo}}{\valOne} \wedge \icsJudgment{\icsOne}{\icontextTwo}{\lcOne}{\circuitTwo}{\lcFour}{\indexThree} \wedge \iLeqJudgment{\itesOne,\icsOne}{}{\indexThree}{\indexOne} \wedge \vcorrect{\itesOne,\icsOne}{\icontextTwo}{\lcFour}{\valOne}{\typeOne}
	\end{equation*}
\end{definition}

\begin{definition}[Correct RMI]
	We say that an RMI $\itesOne$ is \emph{correct} with respect to a CMI $\icsOne$ when every judgment of the form $\itesCJudgment{\itesOne,\icsOne}{\emptyset}{\icontextTwo}{\emptycontext}{\lcOne}{\termOne}{\typeOne}{\indexOne}$ derivable in \PQRA\ is correct.
\end{definition}

For simplicity, we say that the \PQRA\ instance defined by RMI $\itesOne$ and CMI $\icsOne$ is correct when $\itesOne$ is correct with respect to $\icsOne$.  

\subsection{Semantic Interpretation of Types}

In order to prove correctness, we rely on the logical relations technique \cite{logical-relations}.
The first thing we have to do is provide a \emph{semantic interpretation} of types and effects. This interpretation characterizes what it means for a value or term to \emph{semantically} be of a certain type and/or have a certain effect. For the sake of clarity, when reasoning in these terms, we explicitly keep track of which free labels occur in a value or term, by means of a label context $\lcOne$. The formal definition of the semantic interpretation of types and effects is thus as follows.

\begin{definition}[Semantic Interpretation of Types]
	We define $\vsemint{\itesOne,\icsOne}{\icontextTwo}{\typeOne}\subseteq\lcSet\times\valSet$ and $\tsemint{\itesOne,\icsOne}{\icontextTwo}{\typeOne}{\indexOne}\subseteq\lcSet\times\termSet$ as the smallest type-indexed relations such that the following conditions hold:
  \begin{align*}
		(\emptycontext,\unitv) \in \vsemint{\itesOne,\icsOne}{\icontextTwo}{\unitt},\phantom\iff & \\
		(\emptycontext,\nil) \in \vsemint{\itesOne,\icsOne}{\icontextTwo}{\blistt{\ivarOne}{\indexOne}{\typeOne}} \iff &\; \iEqJudgment{}{}{\indexOne}{0},
    \\
		(\labOne:\wtypeOne^{\indexOne},\labOne) \in \vsemint{\itesOne,\icsOne}{\icontextTwo}{\wtypeOne^{\indexTwo}} \iff &\; \iLeqJudgment{\itesOne,\icsOne}{\icontextTwo}{\indexOne}{\indexTwo},
		\\
		(\emptycontext,\lift\termOne) \in \vsemint{\itesOne,\icsOne}{\icontextTwo}{\bang{\indexOne}\typeOne} \iff &\; (\emptycontext,\termOne) \in \tsemint{\itesOne,\icsOne}{\icontextTwo}{\typeOne}{\indexOne},
		\\
		(\lcOne\uplus\lcTwo,\tuple{\valOne}{\valTwo}) \in \vsemint{\itesOne,\icsOne}{\icontextTwo}{\tensor{\typeOne}{\typeTwo}} \iff &\; ((\lcOne,\valOne) \in \vsemint{\itesOne,\icsOne}{\icontextTwo}{\typeOne}) \wedge ((\lcTwo,\valTwo) \in \vsemint{\itesOne,\icsOne}{\icontextTwo}{\typeTwo}),
		\\
		(\lcOne\uplus\lcTwo,\rcons{\valOne}{\valTwo}) \in \vsemint{\itesOne,\icsOne}{\icontextTwo}{\blistt{\ivarOne}{\indexOne}{\typeOne}} \iff &\; \exists \indexTwo. ((\iEqJudgment{}{}{\indexOne}{J+1}) \wedge ((\lcOne,\valOne) \in\vsemint{\itesOne,\icsOne}{\icontextTwo}{\blistt{\ivarOne}{\indexTwo}{\typeOne}})\\
      &\quad\wedge ((\lcTwo,\valTwo)\in\vsemint{\itesOne,\icsOne}{\icontextTwo}{\typeOne\isub{\indexTwo}{\ivarOne}})),
		\\
		(\lcOne,\abs{\varOne}{\typeOne}{\termOne}) \in \vsemint{\itesOne,\icsOne}{\icontextTwo}{\arrowt{\typeOne}{\typeTwo}{\indexOne}{\btypeOne}} \iff &\; ((\lcOne,\labels(\termOne)) \in \vsemint{\itesOne,\icsOne}{\icontextTwo}{\btypeOne}) \wedge \forall\lcTwo,\valOne. (((\lcTwo,\valOne)\in\vsemint{\itesOne,\icsOne}{\icontextTwo}{\typeOne})\\
      &\quad\implies ((\lcOne\uplus\lcTwo,\termOne\sub{\valOne}{\varOne})\in\tsemint{\itesOne,\icsOne}{\icontextTwo}{\typeTwo}{\indexOne})),
		\\
		(\emptycontext,\boxedCirc{\bundle\labOne}{\circuitOne}{\bundle\labTwo}) \in \vsemint{\itesOne,\icsOne}{\icontextTwo}{\mcirct{\indexOne}{\icontextThree}{\btypeOne}{\btypeTwo}} \iff &\; (\icsJudgment{\icsOne}{\icontextThree}{\lcOne}{\circuitOne}{\lcTwo}{\indexTwo})\\
      &\quad\wedge (\iLeqJudgment{\itesOne,\icsOne}{\phantom{}}{\indexTwo}{\indexOne}) \wedge (\iLeqJudgment{\itesOne,\icsOne}{\phantom{}}{\resourcecontent{\lcOne}}{\indexOne}) \wedge (\iLeqJudgment{\itesOne,\icsOne}{\phantom{}}{\resourcecontent{\lcTwo}}{\indexOne})\\
      &\quad\wedge ((\lcOne,\bundle\labOne) \in \vsemint{\itesOne,\icsOne}{\icontextThree}{\btypeOne}) \wedge ((\lcTwo,\bundle\labTwo) \in \vsemint{\itesOne,\icsOne}{\icontextThree}{\btypeTwo}),
		\\
		(\lcOne,\termOne) \in \tsemint{\itesOne,\icsOne}{\icontextTwo}{\typeOne}{\indexOne} \iff &\; \forall\circuitOne. ((\icsJudgment{\icsOne}{\icontextTwo}{\lcTwo}{\circuitOne}{\lcOne,\lcThree}{\indexTwo}) \implies \exists\circuitTwo,\valOne,\lcFour,\indexThree.\\
      &\quad((\config{\circuitOne}{\termOne} \eval \config{\concat{\circuitOne}{\circuitTwo}}{\valOne}) \wedge (\icsJudgment{\icsOne}{\icontextTwo}{\lcOne}{\circuitTwo}{\lcFour}{\indexThree})\\
      &\quad\wedge (\iLeqJudgment{\itesOne,\icsOne}{\phantom{}}{\indexThree}{\indexOne})\wedge (\iLeqJudgment{\itesOne,\icsOne}{\phantom{}}{\resourcecontent{\lcOne}}{\indexOne}) \wedge (\iLeqJudgment{\itesOne,\icsOne}{\phantom{}}{\resourcecontent{\lcFour}}{\indexOne})\\
      &\quad\wedge ((\lcFour,\valOne)\in\vsemint{\itesOne,\icsOne}{\icontextTwo}{\typeOne}))).
	\end{align*}
\end{definition}

Note that because $\lcOne$ and $\typeOne$ may naturally contain free variables encoding local metrics, the interpretation relation is annotated with the index context $\icontextTwo$ from which these variables come from. This specific approach to logical relations, where a semantic interpretation is naturally given to non-ground types, is known as open logical relations \cite{open-logical-relations}.
Predicate $\vsemint{\itesOne,\icsOne}{\icontextTwo}{\typeOne}$ mostly follows the structure of types. Predicate $\tsemint{\itesOne,\icsOne}{\icontextTwo}{\typeOne}{\indexOne}$, on the other hand, tells us what we expect from the execution of a term $\termOne$ of type $\typeOne$ and effect annotation $\indexOne$. Intuitively, it tells us that whenever we pair $\termOne$ with a circuit $\circuitOne$ that supplies the wires in $\lcOne$, $\termOne$ evaluates to $\valOne$ and extends $\circuitOne$ with a subcircuit $\circuitTwo$ with specific properties. Namely, $\circuitTwo$'s size, as well as the size of its inputs and outputs, are upper-bounded by $\indexOne$. Furthermore, the labels output by $\circuitTwo$, i.e. those in $\lcFour$, are all and only the free labels occurring in $\valOne$, and together with these labels $\valOne$ is semantically of type $\typeOne$.

Note that these definitions only apply to expressions that are closed in the sense of values. In order to deal with open expressions, we extend the semantic interpretation of types to typing contexts. Informally, a typing context $\contextOne$ is interpreted by a value substitution $\vsubOne$ which substitutes each $\varOne$ in $\contextOne$ with a value $\valOne$ that is semantically of type $\contextOne(\varOne)$. The formal definition is as follows.

\begin{definition}[Semantic Interpretation of Typing Contexts]
	We define $\csemint{\itesOne,\icsOne}{\icontextTwo}{\contextOne}\subseteq\lcSet\times\vsubSet$ as the smallest context-indexed relation such that the following conditions hold:
	\begin{equation*}
		(\emptycontext,\emptysub) \in \csemint{\itesOne,\icsOne}{\icontextTwo}{\emptycontext},
		\qquad
		(\lcOne\uplus\lcTwo,\vsubOne\extension{\varOne}{\valOne}) \in \csemint{\itesOne,\icsOne}{\icontextTwo}{\contextOne,\varOne:\typeOne} \iff (\lcOne,\vsubOne)\in\csemint{\itesOne,\icsOne}{\icontextTwo}{\contextOne} \wedge (\lcTwo,\valOne)\in\vsemint{\itesOne,\icsOne}{\icontextTwo}{\typeOne},
	\end{equation*}
	where, for all $e \in \valSet \cup \termSet$, $\emptysub(e) = e$ and $\vsubOne\extension{\varOne}{\valOne}(e) = \vsubOne(e\sub{\valOne}{\varOne})$.
\end{definition}

Now we have all the ingredients to define the semantic well-typedness of (possibly open) values and terms. Recall that $\isubOne\in\isemint{\emptyset}{\icontextOne}$ means that index substitution $\isubOne$ substitutes each index variable in $\icontextOne$ with a closed index. We have the following definition.

\begin{definition}[Semantic Well-typedness] If for all $\isubOne\in\isemint{\emptyset}{\icontextOne}$ and all $(\lcTwo,\vsubOne)\in\csemint{\itesOne,\icsOne}{\icontextTwo}{\isubOne(\contextOne)}$ we have
	\begin{enumerate}
		\item $(\lcOne\uplus\lcTwo,\vsubOne(\isubOne(\termOne)))\in \tsemint{\itesOne,\icsOne}{\icontextTwo}{\isubOne(\typeOne)}{\isubOne(\indexOne)}$, then we write $\semCJudgment{\itesOne,\icsOne}{\icontextOne}{\icontextTwo}{\contextOne}{\lcOne}{\termOne}{\typeOne}{\indexOne}$,
		\item $(\lcOne\uplus\lcTwo,\vsubOne(\isubOne(\valOne)))\in \vsemint{\itesOne,\icsOne}{\icontextTwo}{\isubOne(\typeOne)}$, then we write $\semVJudgment{\itesOne,\icsOne}{\icontextOne}{\icontextTwo}{\contextOne}{\lcOne}{\valOne}{\typeOne}$.
	\end{enumerate}
\end{definition}

\subsection{Proving Correctness}

Next, we show that \PQRA\ typing entails semantic typing, and that semantic typing in turn entails correctness. Naturally, this proof would be impossible to carry out without imposing some constraints on the RMI defining the derivable judgments. We therefore identify a small set of constraints on RMIs, which we call \emph{local coherence}, that are sufficient to ensure that an RMI is correct. These constraints relate an RMI to an underlying CMI, and ensure that the former is indeed a sound overapproximation of the latter.

\begin{definition}[Local Coherence]
	We say that an RMI $\itesOne$ is \emph{locally coherent} with a CMI $\icsOne$ when
	\begin{enumerate}
		\item $\icsId{\icsOne}{\emptyset} \leq \itesMempty{\itesOne}$ and $\icsId{\icsOne}{\multiset{\lcOne}+\{\wtypeOne\}} \leq \itesPar{\itesOne}(\icsId{\icsOne}{\multiset{\lcOne}},\itesWires{\itesOne}{\wtypeOne})$, for all $\lcOne\in\lcSet,\wtypeOne\in\{\qubitt,\bitt\}$,
		\item $\icsAppend{\icsOne}{\gateOne}(\itesPar{\itesOne}(\itesWires{\itesOne}{\wtypeOne},n),\itesPar{\itesOne}(\itesWires{\itesOne}{\wtypeOne},m),i,o) \leq \itesPar{\itesOne}(\itesWires{\itesOne}{\wtypeOne},\icsAppend{\icsOne}{\gateOne}(n,m,i,o))$, for all $\gateOne\in\gateSet,\wtypeOne\in\wtypeSet,n,m,i,o\in\natSet$,
		\item $\icsAppend{\icsOne}{\gateOne}(\itesSeq{\itesOne}(p,n),m,i,o) \leq \itesSeq{\itesOne}(p,\icsAppend{\icsOne}{\gateOne}(n,m,i,o))$, for all $\gateOne\in\gateSet,p,n,m,i,o\in\natSet$.
	\end{enumerate}
\end{definition}

The first constraint ensures that the type size function induced by $\itesOne$ correctly overapproximates the size of identity circuits, i.e. that for all $\lcOne\in\lcSet$ we have $\iLeqJudgment{\itesOne,\icsOne}{}{\rcount{\lcOne}}{{\multiset{\lcOne}}}$. Constraint (2) is a sort of lax distributive law, telling us that factoring $\itesPar{\itesOne}$ over $\icsAppend{\icsOne}{\gateOne}$ yields a sound overapproximation. Conversely, constraint (3) tells us that if a gate is appended to the composition in sequence of two circuits $\circuitOne$ and $\circuitTwo$, we can overapproximate the overall size by considering the composition in sequence of $\circuitOne$ and $\gateOne$ appended to $\circuitTwo$.
For simplicity, we say that a \PQRA\ instance is locally coherent when its RMI is locally coherent with its CMI. Under this condition, we can show that well-typedness in \PQRA\ implies semantic well-typedness.

\begin{lemma}
	\label{lem:syntactic-wt-to-semantic-wt}
	Let $\itesOne$ be an RMI which is locally coherent with a CMI $\icsOne$. Let $\Pi$ be a type derivation. We have that
	\begin{align*}
		\itesCJudgment{\itesOne,\icsOne}{\icontextOne}{\icontextTwo}{\contextOne}{\lcOne}{\termOne}{\typeOne}{\indexOne} &\implies \semCJudgment{\itesOne,\icsOne}{\icontextOne}{\icontextTwo}{\contextOne}{\lcOne}{\termOne}{\typeOne}{\indexOne},
		\\
		\itesVJudgment{\itesOne,\icsOne}{\icontextOne}{\icontextTwo}{\contextOne}{\lcOne}{\valOne}{\typeOne} &\implies \semVJudgment{\itesOne,\icsOne}{\icontextOne}{\icontextTwo}{\contextOne}{\lcOne}{\valOne}{\typeOne}.
	\end{align*}
\end{lemma}
\begin{proof}
	By induction on the size of $\Pi$.
\end{proof}

Next, it is easy to show that semantic well-typedness implies correctness both in the sense of global resources and of local resources, by which we conclude the main correctness claim.

\begin{lemma}
	\label{lem:semantic-wt-to-correctness}
	If $\semCJudgment{\itesOne,\icsOne}{\emptyset}{\icontextTwo}{\emptycontext}{\lcOne}{\termOne}{\typeOne}{\indexOne}$, then for every $\circuitOne$ such that $\icsJudgment{\itesOne,\icsOne}{{\icontextTwo}}{\lcTwo}{\circuitOne}{\lcOne,\lcThree}{\indexTwo}$ there exist $\circuitTwo,\valOne,\lcFour,\indexThree$ such that
  \begin{equation*}
		\config{\circuitOne}{\termOne} \eval \config{\concat{\circuitOne}{\circuitTwo}}{\valOne} \wedge \icsJudgment{\itesOne,\icsOne}{\icontextTwo}{\lcOne}{\circuitTwo}{\lcFour}{\indexThree} \wedge \iLeqJudgment{\itesOne,\icsOne}{}{\indexThree}{\indexOne} \wedge \vcorrect{\itesOne,\icsOne}{\icontextTwo}{\lcFour}{\valOne}{\typeOne}.
	\end{equation*}
\end{lemma}
\begin{proof}
	By the definition of $\semCJudgment{\itesOne,\icsOne}{\emptyset}{\icontextTwo}{\emptycontext}{\lcOne}{\termOne}{\typeOne}{\indexOne}$ we get the first three clauses and $(\lcFour,\valOne)\in\vsemint{\itesOne,\icsOne}{\icontextTwo}{\typeOne}$. It remains to show that $(\lcFour,\valOne)\in\vsemint{\itesOne,\icsOne}{\icontextTwo}{\typeOne}$ implies $\vcorrect{\itesOne,\icsOne}{\icontextTwo}{\lcFour}{\valOne}{\typeOne}$, which is easy by induction on the proof of $(\lcFour,\valOne)\in\vsemint{\itesOne,\icsOne}{\icontextTwo}{\typeOne}$.
\end{proof}

\begin{theorem}[Total Correctness]
	Every locally coherent \PQRA\ instance is correct.
\end{theorem}
\begin{proof}
	Let $\itesOne$ be an RMI that is locally coherent with a CMI $\icsOne$ and let $\itesCJudgment{\itesOne,\icsOne}{\emptyset}{\icontextTwo}{\emptycontext}{\lcOne}{\termOne}{\typeOne}{\indexOne}$ be a derivable judgment in \PQRA. By Lemma \ref{lem:syntactic-wt-to-semantic-wt}, since $\emptysub\in\isemint{\emptyset}{\emptyset}$ and $(\emptycontext,\emptysub) \in \csemint{\itesOne,\icsOne}{\icontextTwo}{\emptycontext} = \csemint{\itesOne,\icsOne}{\icontextTwo}{\emptysub(\emptycontext)}$, we get $\semCJudgment{\itesOne,\icsOne}{\emptyset}{\icontextTwo}{\emptycontext}{\lcOne}{\termOne}{\typeOne}{\indexOne}$. In turn, by Lemma \ref{lem:semantic-wt-to-correctness}, $\semCJudgment{\itesOne,\icsOne}{\emptyset}{\icontextTwo}{\emptycontext}{\lcOne}{\termOne}{\typeOne}{\indexOne}$ entails that $\itesCJudgment{\itesOne,\icsOne}{\emptyset}{\icontextTwo}{\emptycontext}{\lcOne}{\termOne}{\typeOne}{\indexOne}$ is correct. Since we imposed no constraints on $\itesCJudgment{\itesOne,\icsOne}{\emptyset}{\icontextTwo}{\emptycontext}{\lcOne}{\termOne}{\typeOne}{\indexOne}$, we get that every judgment of that form derivable in the given \PQRA\ instance is correct, and thus the instance is itself correct.
\end{proof}
\section{Implementing \PQRA\ within \qura}
\label{sec:implementation}

\qura\ is a tool for the analysis of the size of the circuits produced by quantum circuit description programs. \qura\ is implemented in \Haskell\ and is originally based on the work of Colledan and Dal Lago \cite{proto-quipper-r}. As such, it is originally limited to the estimation of only the width of the quantum circuits produced by the programs. More specifically, \qura\ takes as input a program written in a concrete dialect of the \PQ\ language and returns the type of the program, alongside the width of the circuit it builds. \qura\ relies on a type inference algorithm and on SMT-solving \cite{smt-solvers} to infer this information automatically, with minimal program annotations.

Our contribution involves a substantial generalization and extension of \qura's global resource estimation framework, based on the theoretical work described in the previous sections. More in detail, our contribution is twofold: first, we use \PQRA's type-and-effect system as a basis for the generalization of \qura's inference algorithm: at each inference step, instead of synthesizing an index describing \emph{specifically} the width of the constructed circuit, we synthesize an index based on abstract resource composition operations. Whenever we need to perform checks on indices, these operations are translated into different concrete arithmetic operators depending on the definition of the resource under analysis. This definition is independent from the main inference logic, being implemented in an external module as a record of interpretation functions that closely mimics an RMI.

Next, starting from this abstract framework, we reimplemented width estimation as a resource definition module. In addition to width, we also implemented a variety of other resource definition modules, thus extending \qura\ with the ability of analysing global metrics such as gate count, T-count, number of qubits and number of bits.
As a testament to the degree of flexibility guaranteed by \PQRA, the implementation of the aforementioned resource definition modules took less than half an hour to complete overall. The implementation of each of these analyses is a simple \Haskell\ module of less than 40 LOC, including the assignment of ground truth size to elementary gates and operations. The source code of our contribution is available in the \qura\ repository.\footnote{\url{https://github.com/andreacolledan/qura/tree/multimetric-analysis}}

Using our extension of \qura, we were able to analyse a number of quantum algorithms for their global resource consumption, including Grover's \cite{grover} and the Quantum Fourier Transform (QFT) algorithm \cite{qft}, which we use as an example in the next section.

The extension of \qura\ with the ability to analyze \emph{local} metrics is the subject of ongoing efforts. However, because local resource analysis requires more fundamental changes to happen both at the level of \qura's input language and at the level of type inference, progress in this sense is still in the early stages.

\subsection{Global Metrics and the Quantum Fourier Transform, Automatically}
\label{subsec:global-examples}

We take advantage of \qura's implementation of \PQRA\ to illustrate what kind of analysis can be carried out through our generic approach to resource estimation. Consider the program in Listing \ref{lst:qft-program}. This program implements the QFT algorithm, a fundamental subroutine in several quantum algorithms. The programming language used is the input language to \qura, and it is not too dissimilar from the language described in Section \ref{sec:pqra}, with a few notable differences: the syntax is \Haskell-like, indices are explicitly scoped (\texttt{@n.} binds the index \texttt{n} in its syntactical scope) and instantiated (\texttt{(@n.e) @ I} instantiates \texttt{n} with \texttt{I} in \texttt{e}), and index annotations are enclosed within square brackets (e.g. \texttt{List[j<n] Qubit} corresponds to $\blistt{\ivarTwo}{n}{\qubitt}$). Note that no local resource analysis is happening, so wire types are unannotated.
We assume that \texttt{qrev} is a function that reverses a list of qubits, while \texttt{cR @i} applies a controlled rotation (whose magnitude depends on \texttt{i}) to a pair of qubits. The full code of the QFT is available in the repository under \texttt{examples}.

\begin{lstlisting}[language=Haskell, frame=single, style=myStyle, float, label=lst:qft-program,
  caption={An implementation of the QFT algorithm}]
-- Apply the k-th rotation of the m-th iteration
let rotate = @m. lift @k.\((ctrls, trg), ctrl)::((List[j<k] Qubit, Qubit), Qubit).
    let rotation = force cR @ (m+1-k) in 
    let (ctrl, trg) = rotation ctrl trg in
    (ctrls:ctrl, trg)
in
let qft = @n.\reg :: List[j<n] Qubit. -- Apply the QFT to a register of n qubits
  let qftStep = lift @m.\(ctrls, trg)::(List[j<m] Qubit, Qubit). -- m-th iteration
      let revctrls = (force qrev @m) ctrls in
      let (ctrls, trg) = fold(rotate @m, ([], trg), revctrls) in
      let trg = force hadamard trg in
      ctrls:trg
  in fold(qftStep, [], reg)
in qft -- Analyze the function
\end{lstlisting}

When the program is fed to \qura\, we can decide to just type-check it, or to also estimate the global metrics associated with the circuit it builds. For example, if we are interested in the width of the circuits produced by \texttt{qft}, we run

\begin{lstlisting}[language=sh, style=myStyleEmbeddedPlain]
$ qura examples/qft.pqr -g width
Inferred type: n ->[0,0] List[j<n] Qubit -o[n,0] List[j<n] Qubit
\end{lstlisting}

This tells us that the \texttt{qft} function takes as input a list of $n$ qubits and builds a circuit of width at most $n$ which outputs $n$ qubits. On the other hand, if we are interested in the gate count of the circuits produced by \texttt{qft}, we run

\begin{lstlisting}[language=sh, style=myStyleEmbeddedPlain]
$ qura examples/qft.pqr -g gatecount
Inferred type: n ->[0,0] List[j<n] Qubit -o[sum[m<n] m+1, 0] List[j<n] Qubit
\end{lstlisting}

This tells us that on an input of size $n$, the \texttt{qft} function builds a circuit comprising of at most $\sum_{m=0}^{n-1} m+1$ gates. These bounds are not only sound, but also exact, as can be seen in the case of the QFT circuit of input size $3$ shown in Figure \ref{fig:qft-circuit}. Note also that different metrics can be analyzed \emph{without changing the program}, as \qura\ requires few annotations. This is not always the case: higher-order functions, for example, require the programmer to slightly tweak the arrow annotations depending on the underlying resource being analyzed.
\section{Local Metrics and the Quantum Fourier Transform}
\label{sec:examples}
We conclude the exposition of our work by providing an example local resource estimation. Specifically, we use \PQRA\ to perform a pen-and-paper analysis of the \emph{depth} of the QFT program from Section \ref{subsec:global-examples}. Although \qura\ does not yet support the analysis of local resources, in this section we borrow its syntax for ease of exposition and extend it with wire annotations. Let \texttt{Qubit\{i\}} correspond to $\qubitt^\ivarOne$, i.e. the type of a qubit wire at depth $\ivarOne$. The program in Listing \ref{lst:qft-program-local} is then a version of the QFT program in which wire types are decorated, as to allow for the analysis of the depth local metric.

\begin{lstlisting}[language=Haskell, frame=single, style=myStyle, float, label=lst:qft-program-local,
  caption={An implementation of the QFT algorithm supporting local resource analysis}]
let rotate = @i. @m. lift @k.\((ctrls, trg), ctrl)::((List[j<k] Qubit{i+m+j+1}, Qubit{i+m+k}), Qubit{i+m+k}).
    let rotation = force cR @ m+1-k @ i+m+k in 
    let (ctrl, trg) = rotation ctrl trg in -- (Qubit{i+m+k+1}, Qubit{i+m+k+1})
    (ctrls:ctrl, trg) -- (List[j<k+1] Qubit{i+m+j+1}, Qubit{i+m+k+1})
in
let qft = @n. @i.\reg :: List[j<n] Qubit{i}.
  let qftStep = lift @m.\(ctrls, trg)::(List[j<m] Qubit{i+m+j}, Qubit{i}).
      let revctrls = (force qrev @m) ctrls in -- List[j<m] Qubit{i+2m-(j+1)}
      let (ctrls, trg) = fold(rotate @m, ([], trg), revctrls) in
      -- (ctrls, trg) :: (List[j<m] Qubit{i+m+j+1}, Qubit{i+2m})
      let trg = force hadamard trg in -- Qubit{i+2m+1}
      ctrls:trg -- List[i<m+1] Qubit{i+m+1+j}
  in fold(qftStep, [], reg) -- List[j<n] Qubit{i+n+j}
in qft :: n -> i -> List[j<n] Qubit{i} -o List[j<n] Qubit{i+n+j}
\end{lstlisting}

To ease the exposition, we add comments to keep track of the type of key intermediate expressions. We also ignore global resource annotations, and consistently use the same index variable names to denote the same quantities: $n$ is the input size, $i$ is the depth of the input qubits, $m$ is the current QFT iteration, $k$ indexes rotations and $j$ is the position of a qubit within a list.

Despite appearances, most inference steps are straightforward, with the exception of the typing of the two folds. The first fold (line 9) is well typed on account of the fact that $\qubitt^{i+2m-(j+1)}\isub{m-(k+1)}{j}=\qubitt^{i+m+k}$, that $\itesSubJudgment{\mathit{Depth}}{\{i,m\}}{\emptyset}{\qubitt^i}{\qubitt^{i+m} = \qubitt^{i+m+k}\isub{0}{k}}$, and that $k$ increases by $1$ each time \texttt{rotate} is applied. On the other hand, the second and main fold (line 13) is well-typed on the account that $\qubitt^i\isub{n-(m+1)}{j}=\qubitt^i$ and that $m$ increases by one after each application of \texttt{qftStep}.
The result of the analysis is that the \texttt{qft} function has type $\arrowt{\blistt{j}{n}{\qubitt^i}}{\blistt{j}{n}{\qubitt^{i+n+j}}}{}{\unitt}$ for all $i,n\in\natSet$. This means that, given as input a list of $n$ qubits at depth $i$, the QFT circuit outputs a list of $n$ qubits, where the $j$-th qubit is at depth $i+n+j$. Once again, this upper bound is exact, as can be seen in the case of the QFT circuit of Figure \ref{fig:qft-circuit}.

\begin{figure}
	\centering
	\fbox{
	\begin{quantikz}[column sep=2.5em, row sep=2pt]
		\\
		\lstick{$q_0:\qubitt^i$} & &&& \gate{R_3} \wire[r][2]["i+3"{above,pos=0.2}]{a} & \gate{R_2} \wire[r][2]["i+4"{above,pos=0.2}]{a} & \gate{H} & \rstick{$q_2:\qubitt^{i+5}$}\\
		\lstick{$q_1:\qubitt^i$} && \gate{R_2} \wire[r][2]["i+2"{above,pos=0.2}]{a} & \gate{H} \wire[r][2]["i+3"{above,pos=0.2}]{a} && \ctrl{-1} \wire[r][2]["i+4"{above,pos=0.2}]{a} && \rstick{$q_1:\qubitt^{i+4}$}\\
		\lstick{$q_2:\qubitt^i$} & \gate{H}  \wire[r][2]["i+1"{above,pos=0.2}]{a} & \ctrl{-1} \wire[r][2]["i+2"{above,pos=0.2}]{a} && \ctrl{-2} \wire[r][2]["i+3"{above,pos=0.2}]{a} &&& \rstick{$q_0:\qubitt^{i+3}$}\\
	\end{quantikz}
	}
	\caption{The QFT circuit on a register of $3$ qubits. Each wire segment is annotated with its depth.}
	\label{fig:qft-circuit}
\end{figure}
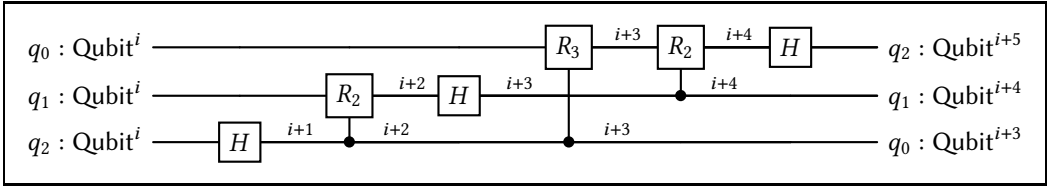
\section{Related Work}
\label{sec:related-work}

\PQRA\ is by no means the first paradigmatic language derived from \Quipper. Some members of the \PQ\ family model the basic features of the language \cite{proto-quipper-s,proto-quipper-m}, while others model its extension with dynamic lifting \cite{proto-quipper-l,proto-quipper-dyn,proto-quipper-k} or dependent types \cite{proto-quipper-d}. The family member that has influenced this work the most is certainly Colledan and Dal Lago's \PQR\ \cite{proto-quipper-r}. As already mentioned several times, however, the work in question deals only with circuit width, a metric which, although very important, is certainly not the only measurement of size.

Resource analysis for quantum programs has been advocated as crucial by the security community, where quantum attacks are considered a threat and it is therefore crucial to understand, for instance, how many qubits are needed to implement certain attacks (see, e.g., \cite{Grassl+2015,Amy+2017}). This type of analysis, however, is usually performed on an individual circuit and not parametrically on the size of the input, which is what our work instead achieves.

The complexity analysis of quantum programs has been the subject of recent investigation, with tools coming for example from the so-called implicit computational complexity \cite{algebra-of-functions,quantum-implicit-complexity}, but also through tools such as weakest precondition calculi \cite{quantum-weakest,quantum-transformer}.
Overall, however, the literature is much more sparse in the quantum case than it is in the classical, deterministic, but even probabilistic case, where existing contributions range from type systems \cite{probabilistic-session-types} to amortized analysis \cite{Hoffmann}, from weakest precondition calculi \cite{runtime-analysis-probabilistic} to abstract interpretation \cite{costa}.
On the other hand, the verification of \emph{functional properties} of quantum programs has received more attention by the research community \cite{quantum-verification-survey}. Various kinds of verification problems have been considered, including termination \cite{DBLP:journals/pacmpl/LiY18,DBLP:journals/acta/LiYY14}, robustness \cite{DBLP:conf/cav/GuanFY21,DBLP:journals/pacmpl/HungHZYHW19}, equivalence checking~\cite{DBLP:journals/jcss/WangLY21}, and compiler optimizations \cite{DBLP:journals/pacmpl/Hietala0HW021} using tools like Hoare logics \cite{hoare-optimization}, interactive theorem proving \cite{DBLP:conf/itp/Hietala0HL021}, and type systems~\cite{quantum-verification-types}. Most of these works are concerned with imperative programs.

On the side of classical programs, it is worth mentioning that refinement types have been succesfully used to perform resource analysis in the context of \Haskell \cite{liquidate-your-assets}.

\section{Conclusion and Future Work}
\label{sec:conclusion}

In this work we presented for the first time a type system capable of deriving upper bounds on the size of the quantum circuits produced by programs in the \Quipper\ language. Crucially, the upper bounds we get are not constant, but rather parametric on classical circuit-building parameters, such as the size of the input. Furthermore, the metric used to measure the size of circuits is not fixed, but can in fact be any metric whose semantic interpretation satisfies some reasonable conditions.

We also used our theoretical framework to massively generalize and extend the capabilities of the \qura\ tool, demonstrating how several parametric global metric analyses can be defined easily and carried out in a fundamentally automatic fashion.

On the side of \qura, benchmarks of the introduced technique, which so far are limited to a modest set of programs and do not consider local circuit metrics, are something we plan to work on in the near future. Specifically, the implementation of local metric analysis into the tool is the object of ongoing efforts, and we believe that results in this sense are not far away.

On the metatheory side, there is certainly an open question regarding whether sufficient conditions exist that can ensure the derivability of a suitable interpretation for a certain metric. In other words, it would be interesting to know that certain circuit metrics can be analyzed in our type system without having to explicitly construct a corresponding CMI or RMI. 
Lastly, a denotational account of the semantics of \PQRA\ is also a topic of future work.

\begin{acks}
The research leading to these results has received funding from the European Union - NextGenerationEU through the Italian Ministry of University and Research under PNRR - M4C2 - I1.4 Project CN00000013 “National Centre for HPC, Big Data and Quantum Computing”
\end{acks}

\bibliographystyle{ACM-Reference-Format}
\bibliography{bibliography}


\end{document}